\documentclass[lettersize,journal]{IEEEtran}
\usepackage{amsmath,amsfonts}
\usepackage{amssymb}
\usepackage{amsthm}
\usepackage{algorithmic}
\usepackage{algorithm}
\usepackage{array}
\usepackage[caption=false,font=normalsize,labelfont=sf,textfont=sf]{subfig}
\usepackage{textcomp}
\usepackage{stfloats}
\usepackage{url}
\usepackage{verbatim}
\usepackage{graphicx}
\usepackage{cite}
\usepackage[colorlinks=true,linkcolor=blue,citecolor=blue,urlcolor=black]{hyperref}

\usepackage{soul}
\usepackage{xcolor}
\usepackage{etoolbox}

\sethlcolor{yellow}

\makeatletter
\renewcommand{\@cite}[2]{%
    \textcolor{blue}{[}%
    \textcolor{black}{#1}%
    \if@tempswa , \textcolor{black}{#2}\fi
    \textcolor{blue}{]}%
}
\renewcommand{\eqref}[1]{\textcolor{blue}{(\ref{#1})}}
\makeatother

\newtheorem{theorem}{Theorem}[section]

\newtheorem{lemma}[theorem]{Lemma}
\newtheorem{proposition}[theorem]{Proposition}

\begin{document}

\title{SINR Maximizing Distributionally Robust Adaptive Beamforming}

\author{Kiarash Hassas Irani, Yongwei Huang,~\IEEEmembership{Senior Member,~IEEE,} and Sergiy A. Vorobyov,~\IEEEmembership{Fellow,~IEEE}

\thanks{Kiarash Hassas Irani and Sergiy A. Vorobyov are with the Department of Information and Communications Engineering,
Aalto University, 02150 Espoo, Finland (e-mail: \href{mailto:kiarash.hassasirani@aalto.fi}{kiarash.hassasirani@aalto.fi}; \href{mailto:sergiy.vorobyov@aalto.fi}{sergiy.vorobyov@aalto.fi}).

Yongwei Huang is with the School of Information Engineering, Guangdong University of Technology, University Town, Guangzhou, Guangdong 510006, China (e-mail: \href{mailto:ywhuang@gdut.edu.cn}{ywhuang@gdut.edu.cn}).}

\thanks{\textit{(Corresponding author: Sergiy A. Vorobyov.)}}
}



\maketitle

\begin{abstract}
This paper addresses the robust adaptive beamforming (RAB) problem via the worst-case signal-to-interference-plus-noise ratio (SINR) maximization over distributional uncertainty sets for the random interference-plus-noise covariance (INC) matrix and desired signal steering vector. Our study explores two distinct uncertainty sets for the INC matrix and three for the steering vector. The uncertainty sets of the INC matrix account for the support and the positive semidefinite (PSD) mean of the distribution, as well as a similarity constraint on the mean. The uncertainty sets for the steering vector consist of the constraints on the first- and second-order moments of its associated probability distribution. The RAB problem is formulated as the minimization of the worst-case expected value of the SINR denominator over any distribution within the uncertainty set of the INC matrix, subject to the condition that the expected value of the numerator is greater than or equal to one for every distribution within the uncertainty set of the steering vector. By leveraging the strong duality of linear conic programming, this RAB problem is reformulated as a quadratic matrix inequality problem. Subsequently, it is addressed by iteratively solving a sequence of linear matrix inequality relaxation problems, incorporating a penalty term for the rank-one PSD matrix constraint. We further analyze the convergence of the iterative algorithm. The proposed robust beamforming approach is validated through simulation examples, which illustrate improved performance in terms of the array output SINR.
\end{abstract}


\section{Introduction}
\IEEEPARstart{B}{eamforming} is an essential technique in array signal processing, widely used in applications such as wireless communications, radar systems, sonar, biomedical applications, acoustics, as well as in various other domains \cite{ShahbazPanahi2018}. It involves the manipulation of signals received by an array of sensors to enhance the reception or transmission in a specific direction, i.e., spatial filtering of signals \cite{Veen1988}. When beamforming weights are optimized using array data or measurements, the process is referred to as adaptive beamforming. This term is used to differentiate it from conventional beamforming, where the weights remain independent of the array data \cite{Vorobyov2014}. However, it is important to note that adaptive beamforming algorithms are prone to performance degradation due to various factors such as mismatch between assumed and actual array response. Robust adaptive beamforming (RAB) techniques have been developed to mitigate these issues, ensuring reliable performance even under adverse conditions \cite{Li2006}.

Various approaches can be employed for RAB designs (see, e.g., \cite{Gershman2010}, \cite{Vorobyov2013} and references therein). In case of minimum variance distortionless response (MVDR) beamformers, studies have focused on two principal approaches: worst-case signal-to-interference-plus-noise ratio (SINR) maximization-based RAB and probability chance constraint RAB. Assuming deterministic steering vector mismatch, techniques based on the worst-case performance aim to enhance the output SINR under the most adverse scenarios, regardless of their likelihood. In contrast, probability chance constraints-based approaches consider statistical model for the signal steering vector (or the steering vector mismatch), \cite{Besson2004}, \cite{Besson2005}, allowing the distortion-less constraint to be maintained with a specified probability. The relation of two approaches is addressed in \cite{Vorobyov2008} in the cases of circularly symmetric Gaussian and worst-case distributions of the steering vector mismatch, and it is shown that probability chance constraint RAB, approximated by convex second-order cone programming (SOCP) problems, can be interpreted as their deterministic worst-case equivalents.

Optimizing the worst-case performance, the authors of \cite{Vorobyov2003} derived a convex formulation for RAB problem based on SOCP, when the mismatch between the actual desired signal steering vector and the presumed steering vector is modeled as a ball centered at the presumed steering vector. The study in \cite{Kim2008} also demonstrates that the worst-case SINR maximization problem can be reformulated as a convex optimization problem, provided the uncertainty sets for both the interference-plus-noise covariance (INC) matrix and the signal steering vector are convex. Specifically, it demonstrates that if the uncertainty model is represented by linear matrix inequalities (LMIs), the problem can be effectively addressed through semidefinite programming (SDP). In case of non-convex uncertainty set for the desired signal steering vector (e.g., the intersection of a ball centered at a nonzero point and a sphere centered at the origin), \cite{Huang2023} shows that SINR maximization problem can be reformulated as a quadratic matrix inequality (QMI) problem. It further introduces a convex method of solving the QMI problem via LMI relaxation and establishes necessary and sufficient conditions for the LMI relaxations to produce rank-one solutions.

Assuming that the probability distribution of the signal steering vector is known, the RAB problem can be formulated as a stochastic programming problem (see, e.g., \cite{Shapiro2014}, \cite{Birge2011}), where expected values are incorporated into the optimization framework. However, the exact knowledge of the underlying distribution might not always be available, and in such cases, distributionally robust optimization (DRO) (see, e.g., \cite{Rahimian2022}, \cite[Chapter 8]{Delage2021}) is employed, which considers uncertainty sets that capture the uncertainty about the true distribution and aims to maximize the worst-case SINR over a range of possible distributions. This technique helps to ensure reliable performance even when the exact statistical properties of the steering vector are unknown. 

The study in \cite{Zhang2015} considers an ellipsoidal constraint on the first-order moment for the signal steering vector, as well as a constraint on the support of the distribution. According to \cite{Zhang2015}, the RAB problem, which involves maximizing the worst-case SINR over the distribution set, can be converted into a convex, tractable form using the S-procedure (also known as the S-lemma) \cite{Fradkov1979}, provided that the support set is composed of a finite number of subsets, each defined by a quadratic constraint. In \cite{Zhang2016}, the chance-constrained RAB is analyzed across various distributional sets, and their performances are compared using numerical simulations. Finally, using the linear conic reformulation of a distributionally robust chance constraint as presented in \cite{Zymler2013}, the distributionally RAB problem is converted into an SDP problem in \cite{Li2018}.

While the aforementioned approaches address robustness by incorporating uncertainty directly into the beamformer design, another line of research focuses on explicitly mitigating model mismatches through the reconstruction of the INC matrix and improved estimation of the steering vector.
For example, the work of \cite{Gu2012} proposes a method to reconstruct the INC matrix via spatial integration of the Capon spectrum over interference-dominated angular sectors. Based on this reconstruction, the presumed steering vector is adaptively corrected to maximize the beamformer output power, under the constraint that the corrected vector must not converge toward any interference source. In \cite{Ruan2014}, a robust beamforming method that estimates the steering vector by projecting a shrinkage-estimated cross-correlation vector between the beamformer output and the array observation onto a predefined signal subspace is proposed. The INC matrix is obtained by subtracting the estimated signal covariance from the data covariance, both computed using the oracle approximating shrinkage (OAS) technique, based on coarse prior knowledge of the signal direction. In addition, computational efficiency of the beamforming has been also addressed. In \cite{Ruan2016}, the steering vector is iteratively refined by projecting the cross-correlation vector onto an orthogonal Krylov subspace. A computationally efficient approach based on power method processing and spatial spectrum matching is presented in \cite{Mohammadzadeh2022}, enabling more accurate reconstruction of the INC matrix and estimation of the desired signal steering vector. 

Several other contributions, including \cite{Huang2015, ZZhang2016, Liao2017, Chakrabarty2018, Yang2018}, further reflect the ongoing interest in steering vector refinement and INC matrix reconstruction, particularly through data-driven and structure-aware approaches aimed at improving robustness in practical RAB scenarios. In addition to enhancing robustness, many of these methods are deliberately designed to improve computational efficiency, prioritizing algorithmic simplicity and lower runtime complexity to better suit real-time or resource-constrained applications.

In this paper,\footnote{Some preliminary outcomes from this research work (covered in Subsection IV.A here) have been initially outlined in \cite{Huang2022}, albeit without detailed proofs.} we investigate the RAB problem through the framework of distributionally robust optimization (DRO) (see, e.g., \cite{Delage2021}). Although some prior works seek to reduce computational complexity, our focus is on formulating a robust approach that provides performance guarantees under distributional uncertainty, without specifically tailoring the algorithm for computational efficiency.

The objective is to minimize the worst-case expected value of the interference-plus-noise power (the denominator of the SINR), considering all possible distributions for the random INC matrix within a specified uncertainty set. Concurrently, we enforce a constraint ensuring that the expected value of the desired signal power (the numerator of the SINR) is no less than one for all distributions of the random steering vector within another uncertainty set. Distinguishing our approach from existing methodologies that typically handle only the uncertainty in the signal steering vector while approximating the INC matrix with the sample data covariance matrix, our DRO model comprehensively addresses distributional uncertainties in both the INC matrix and the steering vector. 

By leveraging the strong duality principle in linear conic programming, we transform the DRO-based RAB problem into a QMI problem with respect to the beamforming vector. This reformulation is achieved under the assumption that the uncertainty set in the objective of the RAB problem comprises a probability measure constraint, a positive semidefinite (PSD) constraint, and a similarity constraint on the first-order moment. Simultaneously, the uncertainty set in the constraint of the RAB problem is defined by first- and second-order moment constraints. The resultant QMI problem is tackled through an LMI relaxation, and we adopt an iterative algorithm to obtain a rank-one solution for the LMI relaxation problem by exploiting the fact that the trace of a nonzero PSD matrix equals its Frobenius norm if and only if it is of rank one. During each iteration, we solve an LMI problem with a penalty term that enforces the rank-one constraint in the objective function. We establish that the stopping criterion of the proposed algorithm can be satisfied, ensuring convergence to a rank-one solution. To demonstrate the efficacy of the proposed method, we present illustrative numerical examples that highlight its superior performance in terms of array output SINR compared to two competitive beamforming techniques. Additionally, the performance of the proposed beamformer is evaluated under different uncertainty sets.
The contributions of this paper are summarized as follows:
\begin{itemize}
\item \textbf{DRO-based RAB Formulation:} The RAB problem is formulated using distributionally robust optimization that accounts for uncertainties in both the INC matrix and the steering vector.
\item \textbf{QMI Problem Reformulation:} The DRO-based RAB problem is reformulated into a QMI problem using strong duality in linear conic programming.
\item \textbf{LMI Relaxation and Rank-one Solution:} An iterative procedure is developed to find a rank-one solution for the LMI relaxation of the QMI problem.
\item \textbf{Evaluation under Various uncertainty Sets:} The performance of the proposed beamformer is analyzed under different uncertainty sets for both the INC matrix and the steering vector.
\item \textbf{Performance Validation:} The performance of the proposed beamformer is evaluated through numerical simulations, demonstrating improved SINR performance over existing competitive beamforming techniques.
\end{itemize}

\subsection{Paper Organization}

The remainder of this paper is organized as follows. In Section~\ref{Signal Model and Problem Formulation}, the signal model is defined, and the problem is formulated, establishing the foundation for the subsequent investigation. Section~\ref{Reformulation of the DRO-based RAB Problem} presents an approach to reformulating the DRO-based RAB problem, highlighting its theoretical framework. A procedure for obtaining a rank-one solution to the LMI relaxation problem is detailed in Section~\ref{A Rank-One Solution Procedure for the LMI Relaxation Problem for the DRO-based RAB}. Section~\ref{Analysis of DRO-based RAB Problem under Alternative uncertainty Sets} explores the DRO-based RAB problem under alternative uncertainty sets. In Section~\ref{Numerical Examples}, the approach is validated through numerical examples, demonstrating its superior performance over existing competitive beamforming techniques. Finally, Section~\ref{Conclusion} concludes the paper by summarizing the findings and contributions.

\subsection{Notation}

Throughout this paper, vectors and matrices are denoted using lowercase and uppercase bold symbols, respectively. The transpose operator is represented by $(\cdot)^T$, and the conjugate transpose operator is denoted by $(\cdot)^H$. The symbol $\boldsymbol{I}$ denotes the identity matrix. Symbols $\boldsymbol{0}$ and $\boldsymbol{1}$ represent matrices or vectors filled with zeros and ones, respectively, with their sizes inferred from the context. The Euclidean norm of a vector is denoted by $\|\cdot\|$ and the Frobenius norm of a matrix is denoted by $\|\cdot\|_F$. The function $\text{tr}(\cdot)$ denotes the trace of a square matrix. The functions $\operatorname{rank}(\cdot)$ and $\lambda_{\max}(\cdot)$ denote the rank and the largest eigenvalue of a matrix, respectively. $\Re (\cdot)$ and $\Im (\cdot)$ represent the real and imaginary parts of a complex-valued argument. The symbol $\succeq$ denotes the generalized inequality: $\boldsymbol{A} \succeq \boldsymbol{B}$ implies that $\boldsymbol{A}-\boldsymbol{B}$ is a Hermitian PSD matrix. The set of all $N \times N$ Hermitian matrices is denoted by $\mathcal{H}^N$, and $\mathcal{H}^N_+$ denotes the set of all Hermitian PSD matrices. In addition, $\mathbb{E}_{G}[\cdot]$ stands for the statistical expectation under the distribution $G$ and $\mathbb{P}_{G}(\cdot)$ represents the probability of an event under the distribution $G$. 
The set of real numbers is denoted by $\mathbb{R}$ and the set of complex numbers by $\mathbb{C}$.

\section{Signal Model and Problem Formulation} \label{Signal Model and Problem Formulation}

Consider the scenario in which an array of $N$ antenna elements receives a point signal source $s(t)$ through the complex-valued steering vector $\boldsymbol{a}$. The received signal of interest (SOI) is therefore formulated as:
\begin{equation}
    \boldsymbol{s}(t) = s(t) \boldsymbol{a}.
\end{equation}
In addition to the desired signal, the observation vector encompasses unwanted components, i.e., noise and interference. Mathematically represented, the equation for the observation vector is:
\begin{equation}
    \boldsymbol{y}(t) = \boldsymbol{s}(t) + \boldsymbol{i}(t) + \boldsymbol{n}(t),
\end{equation}
where $\boldsymbol{s}(t)$, $\boldsymbol{i}(t)$, and $\boldsymbol{n}(t)$ are the statistically independent components of SOI, interference, and noise, respectively. The array output signal is a weighted sum of the signals received at different antenna elements, which is given by:
\begin{equation}
    x(t) = \boldsymbol{w}^H \boldsymbol{y}(t),
\end{equation}
where $\boldsymbol{w}$ is an $N$-dimensional complex-valued vector (i.e., $\boldsymbol{w} \in \mathbb{C} ^N$), representing the beamforming weight vector. Combining the expressions for SOI and the observation vector, the array output signal can be represented as:
\begin{equation}
    x(t) = s(t) \boldsymbol{w}^H \boldsymbol{a} + \boldsymbol{w}^H (\boldsymbol{i}(t) + \boldsymbol{n}(t)).
\end{equation}
Based on this model, we can derive the SINR for the output signal as follows:
\begin{equation}
    \text{SINR} = \frac{\sigma_{\rm s}^2 |\boldsymbol{w}^H \boldsymbol{a}|^2}{\boldsymbol{w}^H \boldsymbol{R}_{\rm{i+n}} \boldsymbol{w}},
\end{equation}
where $\sigma_{\rm s} ^2$ is the power/variance of the source signal and $\boldsymbol{R}_{\rm{i+n}} \triangleq \mathbb{E}[(\boldsymbol{i}(t) + \boldsymbol{n}(t))(\boldsymbol{i}(t) + \boldsymbol{n}(t))^H]$ is the INC matrix. To determine the optimal beamforming weights aligning with the maximum value of the SINR, the optimization problem can be formulated as:
\begin{equation}
    \underset{\boldsymbol{w}}{\text{minimize}} \;\; \boldsymbol{w}^H \boldsymbol{R}_{\rm{i+n}} \boldsymbol{w} \quad \text{subject to} \;\; |\boldsymbol{w}^H \boldsymbol{a}|^2 \geq 1.
\end{equation}
It is worth noting that the constraint of the optimization problem guarantees a distortionless response in the worst-case scenario \cite{Vorobyov2003}.

In real-world scenarios, accessing the INC matrix $\boldsymbol{R}_{\rm{i+n}}$ is typically challenging, and the precise knowledge of the steering vector $\boldsymbol{a}$ is typically unattainable. By considering both $\boldsymbol{R}_{\rm{i+n}}$ and $\boldsymbol{a}$ as random variables (cf. \cite[Chapter 8]{Delage2021}), the DRO-based RAB problem is formulated as:
\begin{equation} \label{eq:DRO-based RAB}
\begin{split}
    \underset{\boldsymbol{w}}{\text{minimize}} \quad &\max _{G_1 \in \mathcal{D}_1} \mathbb{E}_{G_1} [\boldsymbol{w}^H \boldsymbol{R}_{\rm{i+n}} \boldsymbol{w}] \\
    \text{subject to} \quad &\min _{G_2 \in \mathcal{D}_2} \mathbb{E}_{G_2} [\boldsymbol{w}^H \boldsymbol{a} \boldsymbol{a}^H \boldsymbol{w}] \geq 1,
\end{split}
\end{equation}
where $\mathcal{D}_1$ and $\mathcal{D}_2$ represent the sets of the probability distributions of random matrix $\boldsymbol{R}_{\rm{i+n}} \in \mathcal{H}^N$ and random steering vector $\boldsymbol{a}\in \mathbb{C}^N$, respectively.\footnote{Note that, due to the isomorphism between $\mathcal{H}^N$ and $\mathbb{R}^{N^2}$ (see, e.g., \cite{Andersen1995}), the INC matrix $\boldsymbol{R}_{\rm{i+n}}$ can be interpreted as an $N^2$-dimensional real-valued vector. Similarly, $\boldsymbol{a}$ can be viewed as a $2N$-dimensional real-valued vector, formed by concatenating its real and imaginary parts.} Explicitly, the distribution set $\mathcal{D}_1$ is characterized as:
\begin{equation} \label{eq:uncertainty_R}
    \mathcal{D}_1 = \left\{G_1 \in \mathcal{M}_1 \; \; \middle| \; \;
    \begin{array}{l}
        \mathbb{P}_{G_1}(\boldsymbol{R}_{\rm{i+n}} \in \mathcal{Z}_1) = 1\\
        \mathbb{E}_{G_1}[\boldsymbol{R}_{\rm{i+n}}] \succeq \boldsymbol{0}\\
        \|\mathbb{E}_{G_1}[\boldsymbol{R}_{\rm{i+n}}]-\boldsymbol{S}_0\|_F \leq \rho _1\\
    \end{array}
    \right\},
\end{equation}
where $\mathcal{M}_1$ denotes the set of all probability measures on the measurable space $(\mathcal{H}^{N}, \mathcal{B}_1)$, $\mathcal{B}_1$ is the Borel $\sigma$-algebra on $\mathcal{H}^{N}$, $\mathcal{Z}_1 \subseteq \mathcal{H}^{N}$ is a Borel set, and $\boldsymbol{S}_0$ denotes the empirical mean of $\boldsymbol{R}_{\rm{i+n}}$. Specifically, $\boldsymbol{S}_0$ can be approximated using the sample data covariance matrix given by $\widehat{\boldsymbol{R}} = \frac{1}{T} \sum _{t=1} ^T \boldsymbol{y}(t) \boldsymbol{y}^H(t)$,
where $T$ represents the number of available training snapshots. Alternatively, a more accurate approximation of $\boldsymbol{S}_0$ can be obtained using a reconstructed INC matrix, such as the one proposed in \cite{Gu2012}. 

The distribution set $\mathcal{D}_1$ is designed to capture uncertainty while maintaining critical structural properties of the INC matrix $\boldsymbol{R}_{\rm{i+n}}$. Each constraint plays a specific role in ensuring that the uncertainty is well-posed and practically meaningful. By carefully defining these constraints, we ensure that the model remains both robust and reflective of real-world system behavior, accounting for prior knowledge, physical limitations, and potential estimation errors. The key constraints are outlined as follows:

\begin{itemize}
    \item \textbf{Support Constraint:} The constraint $\mathbb{P}_{G_1}(\boldsymbol{R}_{\rm{i+n}}\!\in\!\mathcal{Z}_1)\!=\!1$ ensures that the random matrix $\boldsymbol{R}_{\rm{i+n}}$ always lies within a predefined support set $\mathcal{Z}_1$. This set reflects prior knowledge about the system; for example, physical limitations on power levels of interference and noise. In other words, it defines the range of values that are considered physically reasonable, helping to exclude unrealistic or invalid scenarios.
    
    \item \textbf{PSD Constraint:} The condition $\mathbb{E}_{G_1}[\boldsymbol{R}_{\rm{i+n}}] \succeq \boldsymbol{0}$ guarantees that the expected INC matrix is PSD, since the covariance matrix of a random vector must be PSD under any distribution in order to represent physically meaningful values. Without this constraint, the optimization could result in invalid covariance matrices that do not correspond to realistic scenarios.
    
    \item \textbf{Similarity Constraint:} The constraint $\|\mathbb{E}_{G_1}[\boldsymbol{R}_{\rm{i+n}}] - \boldsymbol{S}_0\|_F \leq \rho_1$ imposes an upper bound on the Frobenius norm of the difference between the expected INC matrix $\mathbb{E}_{G_1}[\boldsymbol{R}_{\rm{i+n}}]$ and the empirical covariance estimate $\boldsymbol{S}_0$.\footnote{Additionally, the Frobenius norm in the constraint could be replaced with other common matrix norms, e.g., the spectral norm.} This similarity constraint ensures that the expected INC matrix $\mathbb{E}_{G_1}[\boldsymbol{R}_{i+n}]$ remains close to a given empirical estimate $\boldsymbol{S}_0$, with the deviation limited by a threshold $\rho_1$. This helps the model stay consistent with observed data while allowing for some uncertainty due to estimation errors or changing conditions. The parameter $\rho_1$ serves as a trust level: a smaller $\rho_1$ indicates strong reliance on the data, whereas a larger value permits greater flexibility to accommodate modeling imperfections.
\end{itemize}

Assuming that the mean vector $\boldsymbol{a}_0 \in \mathbb{C}^N$ and the covariance matrix $\boldsymbol{\Sigma} \succ 0$ of the random vector $\boldsymbol{a}$ under the true distribution $\overline{G}_2$ are known (see, e.g., \cite{Zymler2013}), the uncertainty set $\mathcal{D}_2$, which is a moment-based model, can be characterized as:
\begin{equation} \label{eq:uncertainty_a}
    \mathcal{D}_2 = \left\{G_2 \in \mathcal{M}_2 \; \; \middle| \; \;
    \begin{array}{l}
        \mathbb{P}_{G_2}(\boldsymbol{a} \in \mathcal{Z}_2) = 1\\
        \| \mathbb{E}_{G_2}[\boldsymbol{a}] - \boldsymbol{a}_0 \| \leq \gamma_1\\
        \mathbb{E}_{G_2}[\boldsymbol{a} \boldsymbol{a}^H] \preceq (1+\gamma_2)\boldsymbol{\Sigma} + \boldsymbol{a}_0 \boldsymbol{a}_0 ^H\\
    \end{array}
    \right\},
\end{equation}
where $\mathcal{M}_2$ denotes the set of all probability measures on the measurable space $(\mathbb{C}^{N}, \mathcal{B}_2)$, $\mathcal{B}_2$ is the Borel $\sigma$-algebra on $\mathbb{C}^{N}$, $\mathcal{Z}_2 \subseteq \mathbb{C}^{N}$ is a Borel set. In essence, $\mathcal{D}_2$ encompasses all probability distributions on $\mathcal{Z}_2$ with bounded first- and second-order moments (note that the true second-order moment under $\overline{G}_2$ is $\mathbb{E}_{\overline{G}_2}[\boldsymbol{a} \boldsymbol{a}^H] = \boldsymbol{\Sigma} + \boldsymbol{a}_0 \boldsymbol{a}_0 ^H$).

The distribution set $\mathcal{D}_2$ is formulated to address the inherent uncertainty in the steering vector $\boldsymbol{a}$, ensuring the model remains robust to deviations while preserving key statistical characteristics. 
The key components of this distribution set are as follows:

\begin{itemize}
    \item \textbf{Support Constraint:} Similar to the uncertainty set $\mathcal{D}_1$, the constraint $\mathbb{P}_{G_2}(\boldsymbol{a} \in \mathcal{Z}_2) = 1$ ensures that the random steering vector $\boldsymbol{a}$ always lies within a predefined support set $\mathcal{Z}_2$, reflecting prior knowledge about its feasible values. For instance, $\mathcal{Z}_2$ may encode constraints related to the known angular region from which the signal of interest is expected to arrive, or restrictions on the magnitude and phase of the steering vector elements. This helps the model focus only on physically meaningful signal directions, avoiding unrealistic possibilities.

    \item \textbf{First-Order Moment Constraint:} The constraint $\| \mathbb{E}_{G_2}[\boldsymbol{a}] - \boldsymbol{a}_0 \| \leq \gamma_1$ ensures that the expected value of the random steering vector under any distribution in $\mathcal{D}_2$, lies within a bounded distance $\gamma_1$ from the nominal steering vector $\boldsymbol{a}_0$, the mean of random $\boldsymbol{a}$ under the true distribution $\overline{G}_2$. This mean is assumed to be known \cite{Zymler2013}. By limiting how far the expected value of $\boldsymbol{a}$ can drift from $\boldsymbol{a}_0$, we allow for some uncertainty while still keeping the model centered around realistic and expected signal directions. However, this alone is insufficient to fully capture the distributional uncertainty for the random steering vector, as it lacks higher-order moments’ information, given that the first-order moment of $\boldsymbol{a}$ remains confined to the specified region.
    
    \item \textbf{Second-Order Moment Constraint:} The constraint $\mathbb{E}_{G_2}[\boldsymbol{a} \boldsymbol{a}^H] \preceq (1+\gamma_2)\boldsymbol{\Sigma} + \boldsymbol{a}_0 \boldsymbol{a}_0 ^H$ imposes an upper bound on the second-order moment of $\boldsymbol{a}$, with the level of uncertainty controlled by the parameter $\gamma_2$. This constraint characterizes how the variations of the cross terms of two elements of $\boldsymbol{a}$ are correlated under distribution $G_2 \in \mathcal{D}_2$. In practice, the true steering vector may vary due to noise, interference, or array mismatches. By constraining these variations, the model remains robust to such imperfections, thereby better capturing the uncertainties inherent in real-world scenarios.
\end{itemize}

Higher-order moments, such as skewness (third-order) or kurtosis (fourth-order), are not necessary in this formulation because the uncertainty in the steering vector can be adequately captured by the first- and second-order moments. In beamforming and array signal processing applications, the key sources of distributional uncertainties, i.e., random DOA and steering vector random perturbations, yield typically well-modeled uncertainty sets constrained by the first- and second- order moments \cite{Delage2021}, \cite{Zymler2013}, \cite{Li2018}. Although adding the higher-order moment constraints, which describe more detailed aspects of the distribution shape (such as asymmetry or tail behavior), generally can provide additional information about the distributional uncertainties, this would increase the complexity of the model, without offering substantial practical benefits in the context of robust adaptive beamforming. In other words, the distributional uncertainty sets defined by the first- and second-order moments align with common robust adaptive beamforming frameworks.
Since typical signal and interference conditions do not exhibit significant higher-order effects, the inclusion of higher-order moments would not meaningfully improve the robustness or accuracy of the model.

There are notable differences between the DRO-based RAB problem \eqref{eq:DRO-based RAB} and the distributionally robust beamforming problem discussed in \cite{Li2018}. These differences include:
\begin{itemize}
    \item In addition to the distributional uncertainty of the signal steering vector, problem \eqref{eq:DRO-based RAB} also accounts for the distributional uncertainty associated with the random INC matrix.
    \item The stochastic constraint in \eqref{eq:DRO-based RAB} guarantees that the desired signal power at the array output is at least equal to the source signal power (distortionless under the worst-case scenario) for all distributions within the uncertainty set $\mathcal{D}_2$. In contrast, \cite{Li2018} considers the distributionally robust chance constraint $\mathbb{P}_{P}(\{ \Re(\boldsymbol{w}^H \boldsymbol{a}) \geq 1 \}) \geq p,\, \forall P \in \mathcal{P}$.
    \item In \cite[Eq. (24)]{Li2018}, the distributional uncertainty set is defined based on the moments of the random mismatch vector (i.e., the difference between true and assumed steering vectors), under the assumption that it is a zero-mean vector with uncorrelated components having identical variances. Conversely, the uncertainty set in \eqref{eq:uncertainty_a} is characterized based on the moments of the random signal steering vector without any specific assumptions regarding its mean vector and covariance matrix. This approach provides a broader and more versatile modeling framework.
\end{itemize}

\section{Dual Reformulation of the DRO-based RAB Problem} \label{Reformulation of the DRO-based RAB Problem}

In this section, we derive a dual reformulation of the DRO-based RAB problem \eqref{eq:DRO-based RAB} by replacing the inner maximization and minimization problems, appearing in the objective function and the constraint, respectively, with their corresponding dual formulations. This transformation recasts the original problem into a more tractable form, thereby facilitating further analytical development and enabling efficient numerical implementation.


Let us first deal with the maximization problem in the objective of \eqref{eq:DRO-based RAB}, with the uncertainty set $\mathcal{D}_1$ of distributions as its feasible set, as follows (the subscript of $\boldsymbol{R}_{\rm{i+n}}$ is dropped for notational simplicity):
\begin{equation} \label{eq:objectiveDRO}
\begin{split}
    \underset{G_1 \in \mathcal{M}_1}{\text{maximize}} \quad &\int _{\mathcal{Z}_1} \boldsymbol{w}^H \boldsymbol{R} \boldsymbol{w} \, dG_1(\boldsymbol{R})\\
    \text{subject to} \quad &\int _{\mathcal{Z}_1} dG_1(\boldsymbol{R}) = 1,\\
    &\int _{\mathcal{Z}_1} \boldsymbol{R} \, dG_1(\boldsymbol{R}) \succeq \boldsymbol{0},\\
    &\left \|\int _{\mathcal{Z}_1} \boldsymbol{R} \, dG_1(\boldsymbol{R})-\boldsymbol{S}_0 \right \| _F \leq \rho _1.\\
\end{split}
\end{equation}
The Lagrangian function for problem \eqref{eq:objectiveDRO} is thus expressed as:
\begin{equation} \label{eq:LagrangianObjectiveDRO}
\begin{split}
    \mathcal{L}_1&(G_1, \lambda, \boldsymbol{Y}, \boldsymbol{X}, \mu) = \int _{\mathcal{Z}_1} \boldsymbol{w}^H \boldsymbol{R} \boldsymbol{w} \, dG_1(\boldsymbol{R})\\ 
    &+ \lambda \left(\int _{\mathcal{Z}_1} dG_1(\boldsymbol{R}) - 1 \right) + \text{tr}\left(\boldsymbol{Y}\left(\int _{\mathcal{Z}_1} \boldsymbol{R} \, dG_1(\boldsymbol{R})\right)\right)\\
    &+ \text{tr}\left(\boldsymbol{X}\left(\int _{\mathcal{Z}_1} \boldsymbol{R} \, dG_1(\boldsymbol{R})-\boldsymbol{S}_0\right)\right) + \mu \rho _1,
\end{split}
\end{equation}
where $\lambda$ is the real-valued Lagrange multiplier associated with the first constraint, $\boldsymbol{Y} \succeq \boldsymbol{0}\ (\in \mathcal{H}^N_+)$ is the matrix of Lagrange multipliers associated with the second constraint, and $(\boldsymbol{X}, \mu) \in \mathcal{H}^N \times \mathbb{R}$ with $\|\boldsymbol{X}\|_F \leq \mu$, represents the Lagrange multipliers associated with third constraint. The dual problem for problem \eqref{eq:objectiveDRO} is as follows:
\begin{equation} \label{eq:dualProblemObjective1}
    \underset{\lambda, \boldsymbol{Y}, \boldsymbol{X}, \mu}{\text{inf}} \;\; \underset{G_1 \in \mathcal{M}_1}{\text{sup}} \mathcal{L}_1(G_1, \lambda, \boldsymbol{Y}, \boldsymbol{X}, \mu).
\end{equation}

Next, we determine the supremum of the Lagrangian function \eqref{eq:LagrangianObjectiveDRO}. 
By applying some straightforward mathematical transformations, the Lagrangian function \eqref{eq:LagrangianObjectiveDRO} can be rewritten as (for notational simplicity, $\mathcal{L}_1(G_1, \lambda, \boldsymbol{Y}, \boldsymbol{X}, \mu)$ is henceforth denoted by $\mathcal{L}_1$):
\begin{equation} \label{eq:L1_2}
\begin{split}
    \mathcal{L}_1 &= \mu \rho _1 - \lambda - \text{tr}(\boldsymbol{X}\boldsymbol{S}_0)\\
    &\quad + \int _{\mathcal{Z}_1} (\text{tr} (\boldsymbol{R}(\boldsymbol{w} \boldsymbol{w}^H + \boldsymbol{X} + \boldsymbol{Y})) + \lambda) \, dG_1(\boldsymbol{R}).\\
\end{split}
\end{equation}
To find the supremum of $\mathcal{L}_1$ with respect to $G_1 \in \mathcal{M}_1$, we need to determine the supremum of the integral term in \eqref{eq:L1_2}. If there exists an $\boldsymbol{R} \in  \mathcal{Z}_1$ for which the integrand is positive, the supremum of the integral term becomes $+ \infty$ (see \cite[Chapter 8]{Delage2021}). Therefore, in the dual formulation, the following constraint must be considered:
\begin{equation} \label{eq:constraintLambda}
    - \lambda \geq \text{tr} (\boldsymbol{R}(\boldsymbol{w} \boldsymbol{w}^H + \boldsymbol{X} + \boldsymbol{Y})), \; \forall \boldsymbol{R} \in \mathcal{Z}_1,
\end{equation}
and the supremum of $\mathcal{L}_1$ is as follows: 
\begin{equation}
    \underset{G_1 \in \mathcal{M}_1}{\text{sup}} \mathcal{L}_1 = \mu \rho _1 - \lambda - \text{tr}(\boldsymbol{X}\boldsymbol{S}_0). 
\end{equation}
The constraint \eqref{eq:constraintLambda} is equivalent to:
\begin{equation}
    - \lambda \geq \delta_{\mathcal{Z}_1}(\boldsymbol{w} \boldsymbol{w}^H + \boldsymbol{X} + \boldsymbol{Y}),
\end{equation}
where $\delta_{\mathcal{Z}_1}(\cdot)$ is defined as the support function of $\mathcal{Z}_1$ (see, e.g., \cite[Section 2.4]{Beck2017}). 
Consequently, the dual problem \eqref{eq:dualProblemObjective1} is formulated as:
\begin{equation}
\begin{split}
    \underset{\lambda, \mu, \boldsymbol{X}, \boldsymbol{Y}}{\text{minimize}} \; &\quad \mu \rho _1 - \lambda - \text{tr}(\boldsymbol{X}\boldsymbol{S}_0)\\
    {\text{subject to}} \; & - \lambda \geq \delta_{\mathcal{Z}_1}(\boldsymbol{w} \boldsymbol{w}^H + \boldsymbol{X} + \boldsymbol{Y}),\\
    & \; \mu \geq \|\boldsymbol{X}\|_F\\
    &\; \boldsymbol{X} \in \mathcal{H}^N, \boldsymbol{Y} \succeq \boldsymbol{0}\ (\in \mathcal{H}^N_+), \lambda \in \mathbb{R}.
\end{split} 
\end{equation}
By substituting the minimum value of $-\lambda$ and $\mu$ as defined by the first and second constraints, the equivalent minimization problem is:
\begin{equation} \label{eq:dualProblemObjectiveLast-1}
\begin{split}
    \underset{\boldsymbol{X}, \boldsymbol{Y}}{\text{minimize}} \; &\rho _1 \|\boldsymbol{X}\|_F + \delta_{\mathcal{Z}_1}(\boldsymbol{w} \boldsymbol{w}^H + \boldsymbol{X} + \boldsymbol{Y}) - \text{tr}(\boldsymbol{X} \boldsymbol{S}_0)\\
    {\text{subject to}} \;\; &\boldsymbol{X} \in \mathcal{H}^N, \boldsymbol{Y} \succeq \boldsymbol{0}\ (\in \mathcal{H}^N_+).
\end{split} 
\end{equation}
According to Proposition 3.4 in \cite{Shapiro2001}, strong duality holds if the uncertainty set is non-empty, which is evidently the case for the uncertainty set $\mathcal{D}_1$.

Assuming prior knowledge of an upper bound on the total interference-plus-noise power, the support set $\mathcal{Z}_1$ is defined as $\mathcal{Z}_1 = \{\boldsymbol{R} \in \mathcal{H}^N \; | \; \text{tr}(\boldsymbol{R}) \leq \rho _2, \boldsymbol{R} \succeq \boldsymbol{0}\}$, and the corresponding support function is expressed as (note that in this case, the second constraint in \eqref{eq:uncertainty_R} becomes redundant and can be excluded):
\begin{equation} \label{supportfunctionTraceConstraintSet}
    \delta_{\mathcal{Z}_1}(\boldsymbol{w} \boldsymbol{w}^H + \boldsymbol{X}) = \rho _2 \lambda_{\max}(\boldsymbol{w} \boldsymbol{w}^H + \boldsymbol{X}).
\end{equation}

Therefore, problem \eqref{eq:dualProblemObjectiveLast-1} simplifies to the following convex formulation:
\begin{equation}
\begin{split}
    \underset{\boldsymbol{X}}{\text{minimize}} \; &\rho _1 \|\boldsymbol{X}\|_F + \rho_2 \lambda_{\max}(\boldsymbol{w} \boldsymbol{w}^H + \boldsymbol{X}) - \text{tr}(\boldsymbol{X} \boldsymbol{S}_0)\\
    {\text{subject to}} \;\; &\boldsymbol{X} \in \mathcal{H}^N.
\end{split}  
\end{equation}

In the case where the Frobenius norm of the INC matrix is assumed to be upper-bounded, the set $\mathcal{Z}_1$ can be expressed as $\mathcal{Z}_1 = \{\boldsymbol{R} \in \mathcal{H}^N \; | \; \|\boldsymbol{R}\|_F \leq \rho _2, \boldsymbol{R} \succeq \boldsymbol{0}\}$. Using the Cauchy-Schwarz inequality for the Frobenius inner product, the following identity holds for the support function:
\begin{equation}
    \delta_{\mathcal{Z}_1}(\boldsymbol{w} \boldsymbol{w}^H + \boldsymbol{X}) = \rho _2 \left\|\boldsymbol{w} \boldsymbol{w}^H + \boldsymbol{X} \right\|_F.
\end{equation}

Consequently, problem \eqref{eq:dualProblemObjectiveLast-1} can be rewritten as the following convex formulation:
\begin{equation} \label{eq:dualProblemObjectiveLast}
\begin{split}
    \underset{\boldsymbol{X}}{\text{minimize}} \; &\rho _1 \|\boldsymbol{X}\|_F + \rho_2 \left \|\boldsymbol{w} \boldsymbol{w}^H + \boldsymbol{X} \right\|_F - \text{tr}(\boldsymbol{X} \boldsymbol{S}_0)\\
    {\text{subject to}} \;\; &\boldsymbol{X} \in \mathcal{H}^N.
\end{split} 
\end{equation}

We now proceed to examine the minimization problem in the constraint of \eqref{eq:DRO-based RAB}. Like before, by incorporating the conditions defining the uncertainty set as constraints, the equivalent problem is given by:
\begin{equation} \label{eq:constraintDRO}
\begin{split}
    \underset{G_2 \in \mathcal{M}_2}{\text{minimize}} \quad &\int _{\mathcal{Z}_2} \boldsymbol{a}^H \boldsymbol{w} \boldsymbol{w}^H \boldsymbol{a}  \, dG_2(\boldsymbol{a})\\
    \text{subject to} \quad &\int _{\mathcal{Z}_2} dG_2(\boldsymbol{a}) = 1,\\
    &\left \| \int _{\mathcal{Z}_2} \boldsymbol{a} \, dG_2(\boldsymbol{a}) - \boldsymbol{a}_0 \right \| \leq \gamma_1,\\
    &\int _{\mathcal{Z}_2} \boldsymbol{a} \boldsymbol{a}^H \, dG_2(\boldsymbol{a}) \preceq (1+\gamma_2)\boldsymbol{\Sigma} + \boldsymbol{a}_0 \boldsymbol{a}_0 ^H.\\
\end{split}
\end{equation}
The Lagrangian function for problem \eqref{eq:constraintDRO} is as follows:
\begin{equation} \label{eq:LagrangianConstraintDRO}
\begin{split}
    \mathcal{L}_2&(G_2, x, y, \boldsymbol{x}, \boldsymbol{Z}) = \int _{\mathcal{Z}_2} \boldsymbol{a}^H \boldsymbol{w} \boldsymbol{w}^H \boldsymbol{a} \, dG_2(\boldsymbol{a})\\ 
    &+ x \left(1 - \int _{\mathcal{Z}_2} dG_2(\boldsymbol{a}) \right)\\
    &- \Re \left(\boldsymbol{x}^H \left( \int _{\mathcal{Z}_2} \boldsymbol{a} \, dG_2(\boldsymbol{a}) - \boldsymbol{a}_0 \right)\right) - y \gamma_1\\
    &+ \text{tr} \left( \boldsymbol{Z} \left(\int _{\mathcal{Z}_2} \boldsymbol{a} \boldsymbol{a}^H \, dG_2(\boldsymbol{a}) -  (1+\gamma_2)\boldsymbol{\Sigma} - \boldsymbol{a}_0 \boldsymbol{a}_0^H \right) \right),\\
\end{split}
\end{equation}
where $x \in \mathbb{R}$ is the Lagrange multiplier associated with the first constraint, $(\boldsymbol{x}, y) \in \mathbb{C}^N \times \mathbb{R}$, where $\|\boldsymbol{x}\| \leq y$, represents the Lagrange multipliers associated with the second constraint and $\boldsymbol{Z} \succeq \boldsymbol{0}$ is the matrix of Lagrange multipliers associated with the third constraint. The dual problem for problem \eqref{eq:constraintDRO} is obtained by:
\begin{equation} \label{eq:dualProblemConstraint1}
    \underset{x, y, \boldsymbol{x}, \boldsymbol{Z}}{\text{sup}} \;\; \underset{G_2 \in \mathcal{M}_2}{\text{inf}} \mathcal{L}_2(G_2, x, y, \boldsymbol{x}, \boldsymbol{Z}).
\end{equation}
Through straightforward algebraic operations, the Lagrangian function \eqref{eq:LagrangianConstraintDRO} is rewritten as follows (denoting $\mathcal{L}_2(G_2, x, y, \boldsymbol{x}, \boldsymbol{Z})$ as $\mathcal{L}_2$ for notational simplicity):
\begin{equation} \label{eq:LagrangianConstraintDRO2}
\begin{split}
    \mathcal{L}_2 = &\int _{\mathcal{Z}_2}(\boldsymbol{a}^H (\boldsymbol{w} \boldsymbol{w}^H + \boldsymbol{Z}) \boldsymbol{a} - \Re (\boldsymbol{a}^H \boldsymbol{x}) - x) \, dG_2(\boldsymbol{a})\\
    &+ x + \Re (\boldsymbol{a}_0^H \boldsymbol{x}) - y \gamma_1 - \text{tr}(\boldsymbol{Z}( (1+\gamma_2)\boldsymbol{\Sigma} + \boldsymbol{a}_0 \boldsymbol{a}_0^H)).\\
\end{split}
\end{equation}
In a similar manner, we impose a non-negativity constraint on the integrand within the dual formulation, leading to the infimum of the integral being zero rather than approaching $-\infty$. Therefore, the dual formulation is as follows:
\begin{equation} \label{eq:dualProblemConstraint11}
\begin{split}
    \underset{x, y, \boldsymbol{x}, \boldsymbol{Z}}{\text{maximize}} \;\; &x + \Re (\boldsymbol{a}_0^H \boldsymbol{x}) - y \gamma_1 - \text{tr}(\boldsymbol{Z}( (1+\gamma_2)\boldsymbol{\Sigma} + \boldsymbol{a}_0 \boldsymbol{a}_0^H))\\
    {\text{subject to}} \;\; &\boldsymbol{a}^H (\boldsymbol{w} \boldsymbol{w}^H + \boldsymbol{Z}) \boldsymbol{a} - \Re (\boldsymbol{a}^H \boldsymbol{x}) - x \geq 0, \; \forall \boldsymbol{a} \in \mathcal{Z}_2,\\
    &\|\boldsymbol{x}\| \leq y,\\
    &\boldsymbol{Z} \succeq \boldsymbol{0}, \boldsymbol{x} \in \mathbb{C}^N, x, y \in \mathbb{R}.\\
\end{split} 
\end{equation}
By applying the second constraint on the Lagrange multipliers to the objective function, the problem in \eqref{eq:dualProblemConstraint11} becomes equivalent to the following:
\begin{equation} \label{eq:dualProblemConstraint2}
\begin{split}
    \underset{x, \boldsymbol{x}, \boldsymbol{Z}}{\text{maximize}} \;\; &x + \Re (\boldsymbol{a}_0^H \boldsymbol{x}) - \! \gamma_1 \|\boldsymbol{x}\| \!- \text{tr}(\boldsymbol{Z}((1\!+\!\gamma_2)\boldsymbol{\Sigma} \!+\! \boldsymbol{a}_0 \boldsymbol{a}_0^H))\\
    {\text{subject to}} \;\; &\boldsymbol{a}^H (\boldsymbol{w} \boldsymbol{w}^H + \boldsymbol{Z}) \boldsymbol{a} - \Re (\boldsymbol{a}^H \boldsymbol{x}) - x \geq 0, \; \forall \boldsymbol{a} \in \mathcal{Z}_2,\\
    &\boldsymbol{Z} \succeq \boldsymbol{0}, \boldsymbol{x} \in \mathbb{C}^N, x \in \mathbb{R}.\\
\end{split} 
\end{equation}

The strong duality is similarly ensured due to the non-emptiness of the uncertainty set. Furthermore, the semi-infinite constraint in \eqref{eq:dualProblemConstraint2} can be expressed as:
\begin{equation} \label{eq:dualProblemConstraint_Constraint}
    \left[
    \begin{array}{l}
        \boldsymbol{a}\\
        1
    \end{array}
    \right]^H
    \left[
    \begin{array}{cc}
        \boldsymbol{w} \boldsymbol{w}^H + \boldsymbol{Z} &-\frac{\boldsymbol{x}}{2}\\
        -\frac{\boldsymbol{x}^H}{2} &-x\\
    \end{array}
    \right]
    \left[
    \begin{array}{l}
        \boldsymbol{a}\\
        1
    \end{array}
    \right] \geq 0,  \; \forall \boldsymbol{a} \in \mathcal{Z}_2.
\end{equation}
Similar to the approach outlined in \cite{Zymler2013}, when $\mathcal{Z}_2 = \mathbb{C}^{N}$, \eqref{eq:dualProblemConstraint_Constraint} implies:
\begin{equation} \label{eq:dualProblemConstraint_Constraint2}
    \left[
    \begin{array}{cc}
        \boldsymbol{w} \boldsymbol{w}^H + \boldsymbol{Z} &-\frac{\boldsymbol{x}}{2}\\
        -\frac{\boldsymbol{x}^H}{2} &-x\\
    \end{array}
    \right] \succeq \boldsymbol{0}.
\end{equation}

Additionally, suppose that $\mathcal{Z}_2$ is defined by at most two quadratic constraints. Specifically,
\begin{equation} \label{constraint_quadratic_defined}
    \boldsymbol{a}^H \boldsymbol{Q}_i \boldsymbol{a} + 2 \Re (\boldsymbol{q}_i^H \boldsymbol{a}) + q_i \leq 0, \quad i = 1, \dots, k, 
\end{equation}
where $k \in \{1, 2\}$. According to the $S$-procedure (see, e.g., \cite{Fradkov1979}), if there exists a complex vector $\overline{\boldsymbol{a}}$ in the interior of the set $\mathcal{Z}_2$, the semi-infinite constraint in \eqref{eq:dualProblemConstraint2} can be reformulated as the following QMI:
\begin{equation} \label{constraint_S_lemma}
    \left[
    \begin{array}{cc}
        \boldsymbol{w} \boldsymbol{w}^H + \boldsymbol{Z} &-\frac{\boldsymbol{x}}{2}\\
        -\frac{\boldsymbol{x}^H}{2} &-x\\
    \end{array}
    \right] \succeq \sum _{i=1}^k \tau _i
    \left[
    \begin{array}{cc}
        \boldsymbol{Q}_i &\boldsymbol{q}_i\\
        \boldsymbol{q}_i^H &q_i\\
    \end{array}
    \right],\; \tau _i \leq 0. 
\end{equation}

In the ideal scenario where there is no uncertainty in the steering vector $\boldsymbol{a}$, its norm satisfies $\|\boldsymbol{a}\|^2 = N$. Consequently, a reasonable choice for the support set $\mathcal{Z}_2$ is:
\begin{equation} \label{supportset_steeringvector}
\mathcal{Z}_2 = \{ \boldsymbol{a} \in \mathbb{C}^{N} \; | \; (1-\Delta)N \leq \|\boldsymbol{a}\|^2 \leq (1+\Delta)N \},
\end{equation}
where $0 < \Delta < 1$ quantifies the allowable deviation from the ideal norm. For the quadratic constraints defined in \eqref{supportset_steeringvector}, the QMI constraint in \eqref{constraint_S_lemma} can be rewritten as:
\begin{equation} \label{constraint_S_lemma_2}
\begin{split}
    &\left[
    \begin{array}{cc}
        \boldsymbol{w} \boldsymbol{w}^H + \boldsymbol{Z} &-\frac{\boldsymbol{x}}{2}\\
        -\frac{\boldsymbol{x}^H}{2} &-x\\
    \end{array}
    \right] \succeq\\
    &\tau _1
    \left[
    \begin{array}{cc}
        \boldsymbol{I} &\boldsymbol{0}\\
        \boldsymbol{0} &-(1+\Delta)N\\
    \end{array}
    \right] 
    + \tau _2
    \left[
    \begin{array}{cc}
        -\boldsymbol{I} &\boldsymbol{0}\\
        \boldsymbol{0} &(1-\Delta)N\\
    \end{array}
    \right],\\
    &\tau _1, \tau _2 \leq 0.
\end{split}
\end{equation}

By incorporating the result from \eqref{constraint_S_lemma_2}, along with the dual formulation derived in \eqref{eq:dualProblemObjectiveLast} and \eqref{eq:dualProblemConstraint2} for the inner DRO problems, the DRO-based RAB problem \eqref{eq:DRO-based RAB} can be expressed as:
\begin{equation} \label{eq:DRO-based RAB_last}
\begin{split}
    \text{minimize} \;\; &\rho _1 \|\boldsymbol{X}\|_F + \rho_2 \left\|\boldsymbol{w} \boldsymbol{w}^H + \boldsymbol{X}\right\|_F - \text{tr}(\boldsymbol{X} \boldsymbol{S}_0)\\
    \text{subject to} \;\; &x + \Re (\boldsymbol{a}_0^H \boldsymbol{x}) - \gamma_1 \|\boldsymbol{x}\|\\
    & \qquad \qquad \quad \; - \text{tr}(\boldsymbol{Z}((1+\gamma_2)\boldsymbol{\Sigma} + \boldsymbol{a}_0 \boldsymbol{a}_0^H)) \geq 1,\\
    &\left[
    \begin{array}{cc}
        \boldsymbol{w} \boldsymbol{w}^H + \boldsymbol{Z} &-\frac{\boldsymbol{x}}{2}\\
        -\frac{\boldsymbol{x}^H}{2} &-x\\
    \end{array}
    \right] \succeq\\
    &\hspace{-4pt}\tau _1
    \left[
    \begin{array}{cc}
        \boldsymbol{I} &\boldsymbol{0}\\
        \boldsymbol{0} &-(1+\Delta)N\\
    \end{array}
    \right] 
    + \tau _2
    \left[
    \begin{array}{cc}
        -\boldsymbol{I} &\boldsymbol{0}\\
        \boldsymbol{0} &(1-\Delta)N\\
    \end{array}
    \right],\\
    &\tau _1, \tau _2 \leq 0,\\
    &\boldsymbol{w},\boldsymbol{x} \in \mathbb{C}^N, \; \boldsymbol{X} \in \mathcal{H}^N, \; \boldsymbol{Z} \succeq \boldsymbol{0}, \; x, \tau _1, \tau _2 \in \mathbb{R}.
\end{split}
\end{equation}
Note that here, the semi-infinite constraint is relaxed by being directly imposed on the objective function in \eqref{eq:dualProblemConstraint2}, rather than on its optimal value.

\section{A Rank-One Solution for the LMI Relaxation Problem for the DRO-based RAB} \label{A Rank-One Solution Procedure for the LMI Relaxation Problem for the DRO-based RAB}

The DRO-based RAB formulation in \eqref{eq:DRO-based RAB_last} is non-convex due to the presence of the term $\boldsymbol{w} \boldsymbol{w}^H$, which appears in both the objective function and the second constraint. The term $\boldsymbol{w} \boldsymbol{w}^H$ inherently satisfies $\boldsymbol{w} \boldsymbol{w}^H \succeq \boldsymbol{0}$ and $\operatorname{rank}(\boldsymbol{w} \boldsymbol{w}^H) = 1$. To make the problem tractable, $\boldsymbol{w} \boldsymbol{w}^H$ can be relaxed to $\boldsymbol{W} \succeq \boldsymbol{0}$ by dropping the rank-one constraint, thereby allowing to reformulate problem \eqref{eq:DRO-based RAB_last} as the following convex LMI problem:
\begin{equation} \label{eq:DRO-based RAB_last_relaxed}
\begin{split}
    \text{minimize} \;\; &\rho _1 \|\boldsymbol{X}\|_F + \rho_2 \left\|\boldsymbol{W} + \boldsymbol{X}\right\|_F - \text{tr}(\boldsymbol{X} \boldsymbol{S}_0)\\
    \text{subject to} \;\; &x + \Re (\boldsymbol{a}_0^H \boldsymbol{x}) - \gamma_1 \|\boldsymbol{x}\|\\
    & \qquad \qquad \quad \; - \text{tr}(\boldsymbol{Z}((1+\gamma_2)\boldsymbol{\Sigma} + \boldsymbol{a}_0 \boldsymbol{a}_0^H)) \geq 1,\\
    &\left[
    \begin{array}{cc}
        \boldsymbol{W} + \boldsymbol{Z} &-\frac{\boldsymbol{x}}{2}\\
        -\frac{\boldsymbol{x}^H}{2} &-x\\
    \end{array}
    \right] \succeq\\
    &\hspace{-4pt}\tau _1
    \left[
    \begin{array}{cc}
        \boldsymbol{I} &\boldsymbol{0}\\
        \boldsymbol{0} &-(1+\Delta)N\\
    \end{array}
    \right] 
    + \tau _2
    \left[
    \begin{array}{cc}
        -\boldsymbol{I} &\boldsymbol{0}\\
        \boldsymbol{0} &(1-\Delta)N\\
    \end{array}
    \right],\\
    &\tau _1, \tau _2 \leq 0,\\
    &\boldsymbol{x} \in \mathbb{C}^N, \; \boldsymbol{X} \in \mathcal{H}^N, \; \boldsymbol{W}, \boldsymbol{Z} \succeq \boldsymbol{0}, \; x, \tau _1, \tau _2 \in \mathbb{R}.
\end{split}
\end{equation}

\subsection{An Iterative Algorithm to Obtain a Rank-One Solution}

Although the relaxed problem provides a tractable formulation, the resulting solution $\boldsymbol{W}^\star$ is not guaranteed to be of rank one. To recover a feasible beamforming vector from this relaxed formulation, we need a criterion to first check whether $\boldsymbol{W}^\star$ is of rank one. Should the solution $\boldsymbol{W}^ \star$ to problem \eqref{eq:DRO-based RAB_last_relaxed} be of rank one, specifically $\boldsymbol{W}^ \star = \boldsymbol{w}^ \star {\boldsymbol{w}^ \star}^H$, then $\boldsymbol{w}^ \star$ also serves as an optimal beamforming vector for problem \eqref{eq:DRO-based RAB_last}. In cases where $\boldsymbol{W}^ \star$ is not rank-one, alternative rank-one solutions are sought for \eqref{eq:DRO-based RAB_last_relaxed}. The following lemma provides a criterion based on the equality of the trace and the Frobenius norm of $\boldsymbol{W}$ to determine its rank-one property. This result is particularly useful for introducing a penalty term that enforces the solution to be rank-one.

\begin{lemma} \label{lemma4.1}
    Consider a matrix $\boldsymbol{W} \succeq \boldsymbol{0}$ with $\boldsymbol{W} \neq \boldsymbol{0}$. The condition $\text{tr}(\boldsymbol{W}) = \|\boldsymbol{W}\|_F$ implies that $\boldsymbol{W}$ is of rank one.
\end{lemma}
\begin{proof}
    Let $\lambda _1, \dots, \lambda _N \geq 0$ denote the eigenvalues of the matrix $\boldsymbol{W}$. The condition $\text{tr}(\boldsymbol{W}) = \|\boldsymbol{W}\|_F$ can be equivalently expressed as:
    \begin{equation}
        \left (\sum _{i=1}^N \lambda _i \right)^2 = \sum _{i=1}^N \lambda _i ^2.
    \end{equation}
    This equality holds if and only if there is at most one non-zero eigenvalue. Given that $\boldsymbol{W} \neq \boldsymbol{0}$, there must be exactly one non-zero eigenvalue, implying that $\boldsymbol{W}$ is of rank one.
\end{proof}

Note that the condition $\text{tr}(\boldsymbol{W}) = \|\boldsymbol{W}\|_F$ can be equivalently reformulated as:
\begin{equation} \label{rankOneConstraintReformulation}
    \text{tr}(\boldsymbol{W}) - \frac{\text{tr}(\boldsymbol{W} \boldsymbol{W})}{\|\boldsymbol{W}\|_F} = 0.
\end{equation}
When the solution $\boldsymbol{W}^\star$ is not of rank one, we apply an iterative refinement procedure to encourage a rank-one structure. This is achieved by introducing a penalty term based on the result of Lemma~\ref{lemma4.1}, which quantifies how far a matrix is from being rank-one. The idea is to iteratively minimize this deviation by solving a sequence of LMI problems, each progressively steering the solution closer to a rank-one matrix. Using \eqref{rankOneConstraintReformulation}, we can develop an iterative procedure to find a rank-one solution for \eqref{eq:DRO-based RAB_last_relaxed}. During the $k$-th iteration of this procedure, the following LMI problem is addressed, incorporating a penalty term, which is adaptively updated based on the previous solution, to enforce the rank-one constraint:
\begin{equation} \label{eq:DRO-based RAB_last_relaxed_rankOne}
\begin{split}
    \text{minimize} \;\; &\rho _1 \|\boldsymbol{X}\|_F + \rho_2 \left\|\boldsymbol{W} + \boldsymbol{X}\right\|_F - \text{tr}(\boldsymbol{X} \boldsymbol{S}_0)\\
    &\qquad \qquad \qquad \qquad + \alpha \left( \text{tr}(\boldsymbol{W}) - \frac{\text{tr}(\boldsymbol{W} \boldsymbol{W}_k)}{\|\boldsymbol{W}_k\|_F}\right)\\
    \text{subject to} \;\; &x + \Re (\boldsymbol{a}_0^H \boldsymbol{x}) - \gamma_1 \|\boldsymbol{x}\|\\
    & \qquad \qquad \quad \; - \text{tr}(\boldsymbol{Z}((1+\gamma_2)\boldsymbol{\Sigma} + \boldsymbol{a}_0 \boldsymbol{a}_0^H)) \geq 1,\\
    &\left[
    \begin{array}{cc}
        \boldsymbol{W} + \boldsymbol{Z} &-\frac{\boldsymbol{x}}{2}\\
        -\frac{\boldsymbol{x}^H}{2} &-x\\
    \end{array}
    \right] \succeq\\
    &\hspace{-4pt}\tau _1
    \left[
    \begin{array}{cc}
        \boldsymbol{I} &\boldsymbol{0}\\
        \boldsymbol{0} &-(1+\Delta)N\\
    \end{array}
    \right] 
    + \tau _2
    \left[
    \begin{array}{cc}
        -\boldsymbol{I} &\boldsymbol{0}\\
        \boldsymbol{0} &(1-\Delta)N\\
    \end{array}
    \right],\\
    &\tau _1, \tau _2 \leq 0,\\
    &\boldsymbol{x} \in \mathbb{C}^N, \; \boldsymbol{X} \in \mathcal{H}^N, \; \boldsymbol{W}, \boldsymbol{Z} \succeq \boldsymbol{0}, \; x, \tau _1, \tau _2 \in \mathbb{R}.
\end{split}
\end{equation}
Here, $\alpha > 0$ denotes a predefined penalty parameter chosen sufficiently large to enforce the rank-one constraint. The high-rank solution $\boldsymbol{W}^ \star$, obtained from \eqref{eq:DRO-based RAB_last_relaxed}, is used as the initial point $\boldsymbol{W}_0$ in the iterative procedure detailed in Algorithm~\ref{alg:alg1}. This procedure guarantees convergence to a solution that satisfies the rank-one condition, with increasing accuracy as the iteration progresses. The penalty parameter $\alpha$ plays a key role in enforcing this behavior.

\begin{algorithm}
    \caption{Finding a Rank-One Solution for Problem \eqref{eq:DRO-based RAB_last_relaxed}} \label{alg:alg1}

    \textbf{Input:} $\boldsymbol{S}_0, \boldsymbol{\Sigma}, \boldsymbol{a}_0, \gamma_1, \gamma_2, \rho _1, \rho _2, \alpha, \eta ;$\\
    \textbf{Output:} A rank-one solution $\boldsymbol{w}^ \star {\boldsymbol{w}^ \star}^H$ for problem \eqref{eq:DRO-based RAB_last_relaxed};

    \begin{algorithmic}[1]
        \STATE Solve \eqref{eq:DRO-based RAB_last_relaxed}, returning $\boldsymbol{W}^ \star$
        \IF{$\boldsymbol{W}^ \star$ is of rank one}
            \STATE Output $\boldsymbol{w}^ \star$ with $\boldsymbol{W}^ \star = \boldsymbol{w}^ \star {\boldsymbol{w}^ \star}^H$, and terminate;
        \ENDIF
        \STATE $k \gets 0$; Let $\boldsymbol{W}_k$ be the optimal (high-rank) solution $\boldsymbol{W}^ \star$ for \eqref{eq:DRO-based RAB_last_relaxed};
        \REPEAT
            \STATE Solve the LMI problem \eqref{eq:DRO-based RAB_last_relaxed_rankOne}, obtaining solution $\boldsymbol{W}_{k+1}$;\STATE $k \gets k + 1$;
        \UNTIL{$\text{tr}(\boldsymbol{W}_k) - \frac{\text{tr}(\boldsymbol{W}_k \boldsymbol{W}_{k-1})}{\|\boldsymbol{W}_{k-1}\|_F} \leq \eta$}
        \STATE Output $\boldsymbol{w}^ \star$ with $\boldsymbol{W}_k = \boldsymbol{w}^ \star {\boldsymbol{w}^ \star}^H$.
    \end{algorithmic}
\end{algorithm}

The algorithm solves a sequence of LMI problems, each of the same fixed structure and size as relaxed problem \eqref{eq:DRO-based RAB_last_relaxed}. Although the procedure is iterative, the number of iterations required for convergence is consistently small relative to the problem dimension, which is typical for algorithms that approximate non-convex problems using a sequence of convex problems. Therefore, the additional cost introduced by the iterative nature of the algorithm is limited to a small multiplicative factor, corresponding to the number of sequential convex approximations, which is significantly smaller than the problem size and thus negligible for the overall complexity. This means that the overall computational complexity remains on the same order as that of solving a single LMI problem using interior-point methods, as commonly described in the literature (see, e.g., \cite{Nesterov1994}).

Let us investigate how the termination criterion in Step 9 of Algorithm~\ref{alg:alg1} ensures that $\boldsymbol{W}_k$ is a rank-one solution for \eqref{eq:DRO-based RAB_last_relaxed_rankOne}. The termination criterion implies:
\begin{equation} \label{terminationCriterion}
    \text{tr}(\boldsymbol{W}_k) \approx \frac{\text{tr}(\boldsymbol{W}_k \boldsymbol{W}_{k-1})}{\|\boldsymbol{W}_{k-1}\|_F}.
\end{equation}
According to the Cauchy-Schwarz inequality for the Frobenius inner product, the following inequality holds for the right-hand side (RHS) of \eqref{terminationCriterion}:
\begin{equation} \label{terminationCriterionInequality1}
    \frac{\text{tr}(\boldsymbol{W}_k \boldsymbol{W}_{k-1})}{\|\boldsymbol{W}_{k-1}\|_F} \leq \|\boldsymbol{W}_k\|_F.
\end{equation}
Furthermore, since $\boldsymbol{W}_k$ is a PSD matrix, the left-hand side (LHS) of \eqref{terminationCriterion} satisfies the following inequality:
\begin{equation} \label{terminationCriterionInequality2}
    \|\boldsymbol{W}_k\|_F \leq \text{tr}(\boldsymbol{W}_k).
\end{equation}
Referring to \eqref{terminationCriterionInequality1} and \eqref{terminationCriterionInequality2}, the expression in the termination criterion evaluates to a non-negative value. By combining \eqref{terminationCriterion}, \eqref{terminationCriterionInequality1}, and \eqref{terminationCriterionInequality2}, we conclude that $\text{tr}(\boldsymbol{W}_k) \approx \|\boldsymbol{W}_k\|_F$, implying that $\boldsymbol{W}_k$ is a rank-one solution, based on Lemma~\ref{lemma4.1}.

\subsection{Convergence Analysis of the Iterative Algorithm}

To ensure the convergence of Algorithm~\ref{alg:alg1}, we need to demonstrate the occurrence of the termination criterion. Prior to that, we analyze the optimal values obtained in Algorithm~\ref{alg:alg1}. Suppose that in Step 7 of Algorithm~\ref{alg:alg1}, both an optimal solution $\boldsymbol{W}_{k+1}$ and the optimal value $v_{k+1}$ are obtained for $k=0,1,2,\ldots$. We aim to prove that the sequence $\{v_k\}_k$ of optimal values is non-increasing.

\begin{proposition}\label{nonincreasingOptimalValues}
The sequence of optimal values obtained by Algorithm~\ref{alg:alg1} is non-increasing, specifically,
\begin{equation}
    v_1 \geq v_2 \geq v_3 \geq \cdots.
\end{equation}
\end{proposition}


\begin{proof}
Let the optimal value for problem \eqref{eq:DRO-based RAB_last_relaxed_rankOne} be expressed as
\begin{multline} \label{objective_f1_f2}
    v_{k+1} = f_1(\boldsymbol{W}_{k+1}, \boldsymbol{X}_{k+1}) + \alpha f_{2,k}(\boldsymbol{W}_{k+1}),\\
    k=0,1,2,\ldots,
\end{multline}
where
\begin{equation}\label{objective_f1}
\begin{split} 
    &f_1(\boldsymbol{W}_{k+1}, \boldsymbol{X}_{k+1})\\ 
    &= \rho _1 \|\boldsymbol{X}_{k+1}\|_F + \rho _2 \|\boldsymbol{W}_{k+1} + \boldsymbol{X}_{k+1}\|_F - \text{tr}(\boldsymbol{X}_{k+1} \boldsymbol{S}_0),\\
\end{split}
\end{equation}
and
\begin{equation}\label{objective_f2}
    f_{2,k}(\boldsymbol{W}_{k+1}) = \text{tr}(\boldsymbol{W}_{k+1}) - \frac{\text{tr}(\boldsymbol{W}_{k+1} \boldsymbol{W}_k)}{\|\boldsymbol{W}_k\|_F}.
\end{equation}

Observe that the following inequality
\begin{equation}\label{inequality_f2}
    f_{2,k}(\boldsymbol{W}_{k+1}) \geq f_{2,k+1}(\boldsymbol{W}_{k+1}), \quad k=0,1,2,\ldots,
\end{equation}
implies that
{\small \begin{equation}\label{inequality_f2_equivalent}
    \text{tr}(\boldsymbol{W}_{k+1}) - \frac{\text{tr}(\boldsymbol{W}_{k+1} \boldsymbol{W}_k)}{\|\boldsymbol{W}_k\|_F} \geq \text{tr}(\boldsymbol{W}_{k+1}) - \frac{\text{tr}(\boldsymbol{W}_{k+1} \boldsymbol{W}_{k+1})}{\|\boldsymbol{W}_{k+1}\|_F},
\end{equation}
}\noindent 
which is equivalent to the Cauchy-Schwarz inequality for the Frobenius inner product of $\boldsymbol{W}_k$ and $\boldsymbol{W}_{k+1}$:
\begin{equation}\label{inequality_f2_equivalent_2}
    \|\boldsymbol{W}_{k+1}\|_F \|\boldsymbol{W}_k\|_F \geq \text{tr}(\boldsymbol{W}_{k+1} \boldsymbol{W}_k).
\end{equation}

To show $v_{k+1} \geq v_{k+2}$ for $k=0,1,2,\ldots$, we note that
\begin{equation}
\begin{split}
    v_{k+1} &= f_1(\boldsymbol{W}_{k+1}, \boldsymbol{X}_{k+1}) + \alpha f_{2,k}(\boldsymbol{W}_{k+1}) \\
    &\ge f_1(\boldsymbol{W}_{k+1}, \boldsymbol{X}_{k+1}) + \alpha f_{2,k+1}(\boldsymbol{W}_{k+1}) \\
    &\ge f_1(\boldsymbol{W}_{k+2}, \boldsymbol{X}_{k+2}) + \alpha f_{2,k+1}(\boldsymbol{W}_{k+2}) \\
    &= v_{k+2},
\end{split}
\end{equation}
where the first inequality follows from \eqref{inequality_f2}, and the second inequality holds because $\boldsymbol{W}_{k+2}$ is an optimal solution and $\boldsymbol{W}_{k+1}$ is a feasible point for problem \eqref{eq:DRO-based RAB_last_relaxed_rankOne}, with $\boldsymbol{W}_k$ therein replaced by $\boldsymbol{W}_{k+1}$. Thus, the proof is complete.
\end{proof}

Given that the sequence $\{v_k\}_k$ is non-increasing and bounded below by the optimal value of \eqref{eq:DRO-based RAB_last_relaxed}, its convergence is guaranteed by the monotone convergence theorem. Consequently, the inequality in \eqref{inequality_f2_equivalent_2} becomes an equality at the convergence; otherwise, a strictly smaller optimal value would exist, contradicting the established lower bound. This, in turn, implies the existence of an index $k_0$ such that, for all $k \geq k_0$, the relation $\boldsymbol{W}_{k+1} = \beta_k \boldsymbol{W}_k$ holds. Since $\boldsymbol{W}_{k+1}$ is constrained to be a PSD matrix, it follows that $\beta_k > 0$. Furthermore, for sufficiently large values of $\alpha$, such that the penalty term dominates the objective, the optimizer enforces $\beta_k < 1$, effectively reducing the value of $f_{2,k}(\boldsymbol{W}_{k+1})$ at the expense of increasing $f_1(\boldsymbol{W}_{k+1}, \boldsymbol{X}_{k+1})$. Specifically, for two consecutive iterations, the following relation holds:
\begin{equation} \label{f2_convergence}
    f_{2,k+1}(\boldsymbol{W}_{k+2}) = \beta_{k+1} f_{2,k}(\boldsymbol{W}_{k+1}).\\
\end{equation}
Noting that the penalty term is non-negative for PSD matrices; \eqref{f2_convergence}, with $0 < \beta_{k+1} < 1$, implies that the penalty term converges to zero.

\section{Analysis of DRO-based RAB Problem under Alternative Uncertainty Sets} \label{Analysis of DRO-based RAB Problem under Alternative uncertainty Sets}

In this section, we delve into the analysis of problem \eqref{eq:DRO-based RAB} considering different uncertainty sets for the probability distributions of the random INC matrix $\boldsymbol{R}_{\rm{i+n}}$ and the random steering vector $\boldsymbol{a}$. First, we examine the problem using two alternative uncertainty sets for the probability distribution of the random steering vector. Then, we introduce an alternative uncertainty set for the probability distribution of the random INC matrix.

\subsection{Alternative Uncertainty Sets for Steering Vector}

As an alternative to the uncertainty set $\mathcal{D}_2$, we consider the uncertainty set $\mathcal{D}_2^\prime$, defined as follows:
{\small
\begin{equation} \label{eq:uncertainty_a_prime}
    \mathcal{D}_2^\prime = \left\{G_2 \in \mathcal{M}_2 \; \; \middle| \;
    \begin{array}{l}
        \mathbb{P}_{G_2}(\boldsymbol{a} \in \mathcal{Z}_2) = 1\\
        (\mathbb{E}_{G_2}[\boldsymbol{a}] - \bar{\boldsymbol{a}})^H \boldsymbol{Q}^{-1} (\mathbb{E}_{G_2}[\boldsymbol{a}] - \bar{\boldsymbol{a}}) \leq \gamma _1\\
        \| \mathbb{E}_{G_2}[(\boldsymbol{a} - \bar{\boldsymbol{a}})(\boldsymbol{a} - \bar{\boldsymbol{a}})^H] - \bar{\boldsymbol{\Sigma}} \|_F \leq \gamma _2\\
    \end{array}
    \right\},
\end{equation}
}\noindent 
where the second constraint ensures that the mean of the random vector $\boldsymbol{a}$ resides within an ellipsoid of size $\gamma_1$, centered at the empirical mean $\bar{\boldsymbol{a}}$, with the ellipsoid's shape defined by the positive definite matrix $\boldsymbol{Q} \succ \boldsymbol{0}$. The third constraint in \eqref{eq:uncertainty_a_prime} forces that $\mathbb{E}_{G_2}[(\boldsymbol{a} - \bar{\boldsymbol{a}})(\boldsymbol{a} - \bar{\boldsymbol{a}})^H]$, which is called the centered second-order moment matrix of $\boldsymbol{a}$ (see \cite{Delage2021}) and is different from the covariance matrix unless $\mathbb{E}_{G_2}[\boldsymbol{a}] = \bar{\boldsymbol{a}}$ (e.g., when $\gamma _1=0$), is located in a ball of radius $\gamma _2$, centered at the empirical covariance matrix $\bar{\boldsymbol{\Sigma}}$.

Using the uncertainty set $\mathcal{D}_2^\prime$, the minimization problem within the constraint of problem \eqref{eq:DRO-based RAB} can be reformulated as: 
\begin{equation} \label{eq:constraintDRO_D2_prime}
\begin{split}
    \underset{G_2 \in \mathcal{M}_2}{\text{minimize}} \quad &\int _{\mathcal{Z}_2} \boldsymbol{a}^H \boldsymbol{w} \boldsymbol{w}^H \boldsymbol{a}  \, dG_2(\boldsymbol{a})\\
    \text{subject to} \quad &\int _{\mathcal{Z}_2} dG_2(\boldsymbol{a}) = 1,\\
    &\left \| \int _{\mathcal{Z}_2} \boldsymbol{Q}^{-\frac{1}{2}} \boldsymbol{a} \, dG_2(\boldsymbol{a}) - \boldsymbol{Q}^{-\frac{1}{2}} \bar{\boldsymbol{a}} \right \| \leq \sqrt{\gamma _1},\\
    &\left\| \int _{\mathcal{Z}_2} (\boldsymbol{a} - \bar{\boldsymbol{a}})(\boldsymbol{a} - \bar{\boldsymbol{a}})^H \, dG_2(\boldsymbol{a}) - \bar{\boldsymbol{\Sigma}} \right \|_F \leq \gamma _2.\\
\end{split}
\end{equation}
To find the dual formulation of problem \eqref{eq:constraintDRO_D2_prime}, we need to determine the infimum of the Lagrangian, similar to the approach taken for problem \eqref{eq:constraintDRO}. The Lagrangian function for problem \eqref{eq:constraintDRO_D2_prime} is:
\begin{equation} \label{eq:LagrangianConstraintDRO_D2_prime}
\begin{split}
    \mathcal{L}_2&(G_2, x, y, z, \boldsymbol{z}, \boldsymbol{Z}) = \int _{\mathcal{Z}_2} \boldsymbol{a}^H \boldsymbol{w} \boldsymbol{w}^H \boldsymbol{a} \, dG_2(\boldsymbol{a})\\ 
    &+ x \left( \int _{\mathcal{Z}_2} dG_2(\boldsymbol{a}) - 1 \right)\\
    &- \Re \left(\boldsymbol{z}^H \left( \int _{\mathcal{Z}_2} \boldsymbol{Q}^{-\frac{1}{2}} \boldsymbol{a} \, dG_2(\boldsymbol{a}) - \boldsymbol{Q}^{-\frac{1}{2}} \bar{\boldsymbol{a}} \right)\right) - y \sqrt{\gamma _1}\\
    &- \text{tr} \left( \boldsymbol{Z} \left( \bar{\boldsymbol{\Sigma}} - \int _{\mathcal{Z}_2} (\boldsymbol{a} - \bar{\boldsymbol{a}})(\boldsymbol{a} - \bar{\boldsymbol{a}})^H \, dG_2(\boldsymbol{a}) \right) \right) - z \gamma _2.\\
\end{split}
\end{equation}
Here, $x \in \mathbb{R}$ is the Lagrange multiplier associated with the first constraint, $(\boldsymbol{z}, y) \in \mathbb{C}^N \times \mathbb{R}$, and $(\boldsymbol{Z}, z) \in \mathcal{H}^N \times \mathbb{R}$, where $\|\boldsymbol{z}\| \leq y$, and $\|\boldsymbol{Z}\|_F \leq z$, represent the Lagrange multipliers associated with the second and third constraints, respectively. The Lagrangian \eqref{eq:LagrangianConstraintDRO_D2_prime} can be rewritten as:
\begin{equation} \label{eq:L2_prime}
\begin{split}
    &\mathcal{L}_2 = -x + \Re (\bar{\boldsymbol{a}}^H \boldsymbol{Q}^{-\frac{1}{2}} \boldsymbol{z}) - y \sqrt{\gamma _1} - \text{tr}(\boldsymbol{Z}(\bar{\boldsymbol{\Sigma}} - \bar{\boldsymbol{a}} \bar{\boldsymbol{a}}^H)) - z \gamma _2\\
    &\!\!+\!\! \int _{\mathcal{Z}_2} \!\!(\boldsymbol{a}^H \!(\boldsymbol{w} \boldsymbol{w}^H \!+\! \boldsymbol{Z}) \boldsymbol{a} \!-\! 2 \Re (\boldsymbol{a}^H \!(\boldsymbol{Z} \bar{\boldsymbol{a}} \!+\! \frac{1}{2} \boldsymbol{Q}^{-\frac{1}{2}} \boldsymbol{z})) \!+\! x) \, dG_2(\boldsymbol{a}).\\
\end{split}
\end{equation}
Using similar methodology as in the previous dual problem derivations, 
the dual problem for \eqref{eq:constraintDRO_D2_prime} is formulated as:
\begin{equation} \label{eq:dualProblemConstraint_D2_prime}
\begin{split}
    \underset{x, \boldsymbol{z}, \boldsymbol{Z}}{\text{maximize}} \;\; &-x + \Re (\bar{\boldsymbol{a}}^H \boldsymbol{Q}^{-\frac{1}{2}} \boldsymbol{z}) - \sqrt{\gamma _1}\|\boldsymbol{z}\| - \gamma _2 \|\boldsymbol{Z}\|_F\\
    &\qquad \qquad \qquad \qquad \; \, \, - \text{tr}(\boldsymbol{Z}(\bar{\boldsymbol{\Sigma}} - \bar{\boldsymbol{a}} \bar{\boldsymbol{a}}^H)) \\
    {\text{subject to}} \;\; &\boldsymbol{a}^H \!(\boldsymbol{w} \boldsymbol{w}^H \!\!+\! \boldsymbol{Z}) \boldsymbol{a} \!-\! 2 \Re (\boldsymbol{a}^H \!(\boldsymbol{Z} \bar{\boldsymbol{a}} \!+\! \frac{1}{2} \boldsymbol{Q}^{-\frac{1}{2}} \boldsymbol{z})) \!+\! x \!\geq\! 0,\\
    &\qquad \qquad \qquad \qquad \qquad \qquad \qquad \qquad \; \; \; \forall \boldsymbol{a} \in \mathcal{Z}_2,\\
    &\boldsymbol{Z} \in \mathcal{H}^N, \boldsymbol{z} \in \mathbb{C}^N, x \in \mathbb{R}.\\
\end{split} 
\end{equation}
In addition, strong duality between \eqref{eq:constraintDRO_D2_prime} and \eqref{eq:dualProblemConstraint_D2_prime} holds as the uncertainty set $\mathcal{D}_2^\prime$ is not an empty set.

We remark that by setting $\boldsymbol{x} = \frac{1}{2} \boldsymbol{Q}^{-\frac{1}{2}} \boldsymbol{z}$, problem \eqref{eq:dualProblemConstraint_D2_prime} can be reformulated as follows:
\begin{equation} \label{eq:dualProblemConstraint_D2_prime_2}
\begin{split}
    \underset{x, \boldsymbol{z}, \boldsymbol{Z}}{\text{maximize}} \;\; &-x + 2 \Re (\bar{\boldsymbol{a}}^H \boldsymbol{x}) - 2 \sqrt{\gamma _1} \| \boldsymbol{Q}^{\frac{1}{2}} \boldsymbol{x} \| - \gamma _2 \| \boldsymbol{Z} \|_F\\
    &\qquad \qquad \qquad \ \ \   - \text{tr}(\boldsymbol{Z}(\bar{\boldsymbol{\Sigma}} - \bar{\boldsymbol{a}} \bar{\boldsymbol{a}}^H))\\
    {\text{subject to}} \;\; &\boldsymbol{a}^H (\boldsymbol{w} \boldsymbol{w}^H + \boldsymbol{Z}) \boldsymbol{a} - 2 \Re (\boldsymbol{a}^H (\boldsymbol{Z} \bar{\boldsymbol{a}} + \boldsymbol{x})) + x \geq 0, \\
    &\qquad \qquad \qquad \qquad \qquad \qquad \qquad \quad \quad \; \forall \boldsymbol{a} \in \mathcal{Z}_2,\\
    &\boldsymbol{Z} \in \mathcal{H}^N, \boldsymbol{x} \in \mathbb{C}^N, x \in \mathbb{R}.\\
\end{split} 
\end{equation}

Suppose that $\mathcal{Z}_2$ is defined by \eqref{supportset_steeringvector}. Hence, the semi-infinite constraint in problem \eqref{eq:dualProblemConstraint_D2_prime_2} can be further specified to
\begin{equation} \label{eq:dualProblemConstraint_D2_prime_2_Constraint}
\begin{split}
    &\left[
    \begin{array}{cc}
        \boldsymbol{w} \boldsymbol{w}^H + \boldsymbol{Z} &-\boldsymbol{Z} \bar{\boldsymbol{a}} - \boldsymbol{x}\\
        - \bar{\boldsymbol{a}}^H \boldsymbol{Z} - \boldsymbol{x}^H &x\\
    \end{array}
    \right] \succeq\\
    &\tau _1
    \left[
    \begin{array}{cc}
        \boldsymbol{I} &\boldsymbol{0}\\
        \boldsymbol{0} &-(1+\Delta)N\\
    \end{array}
    \right] 
    + \tau _2
    \left[
    \begin{array}{cc}
        -\boldsymbol{I} &\boldsymbol{0}\\
        \boldsymbol{0} &(1-\Delta)N\\
    \end{array}
    \right],\\
    &\tau _1, \tau _2 \leq 0.
\end{split}
\end{equation}
Accordingly, similar to problem \eqref{eq:DRO-based RAB_last}, the original DRO-based RAB problem \eqref{eq:DRO-based RAB} can be recast into
\begin{equation} \label{eq:DRO-based RAB_last_D2_prime}
\begin{split}
    \text{minimize} \;\; &\rho _1 \|\boldsymbol{X}\|_F + \rho_2 \left\|\boldsymbol{w} \boldsymbol{w}^H + \boldsymbol{X}\right\|_F - \text{tr}(\boldsymbol{X} \boldsymbol{S}_0)\\
    \text{subject to} \;\; &-x + 2 \Re (\bar{\boldsymbol{a}}^H \boldsymbol{x}) - \text{tr}(\boldsymbol{Z}(\bar{\boldsymbol{\Sigma}} - \bar{\boldsymbol{a}} \bar{\boldsymbol{a}}^H))\\
    &\qquad \qquad \qquad \quad \geq 1 + 2 \sqrt{\gamma _1} \| \boldsymbol{Q}^{\frac{1}{2}} \boldsymbol{x} \| + \gamma _2 \| \boldsymbol{Z} \|_F,\\
    &\left[
    \begin{array}{cc}
        \boldsymbol{w} \boldsymbol{w}^H + \boldsymbol{Z} &-\boldsymbol{Z} \bar{\boldsymbol{a}} - \boldsymbol{x}\\
        - \bar{\boldsymbol{a}}^H \boldsymbol{Z} - \boldsymbol{x}^H &x\\
    \end{array}
    \right] \succeq\\
    &\hspace{-4pt}\tau _1
    \left[
    \begin{array}{cc}
        \boldsymbol{I} &\boldsymbol{0}\\
        \boldsymbol{0} &-(1+\Delta)N\\
    \end{array}
    \right] 
    + \tau _2
    \left[
    \begin{array}{cc}
        -\boldsymbol{I} &\boldsymbol{0}\\
        \boldsymbol{0} &(1-\Delta)N\\
    \end{array}
    \right],\\
    &\tau _1, \tau _2 \leq 0.\\
    &\boldsymbol{w},\boldsymbol{x} \in \mathbb{C}^N, \; \boldsymbol{X}, \boldsymbol{Z} \in \mathcal{H}^N, \; x \in \mathbb{R}.
\end{split}
\end{equation}

Another alternative of $\mathcal{D}_2$ is the following uncertainty set:
{\small
\begin{equation} \label{eq:uncertainty_a_prime_prime}
    \mathcal{D}_2^{\prime\prime} = \left\{G_2 \in \mathcal{M}_2 \; \; \middle| \,
    \begin{array}{l}
        \mathbb{P}_{G_2}(\boldsymbol{a} \in \mathcal{Z}_2) = 1\\
        (\mathbb{E}_{G_2}[\boldsymbol{a}] - \bar{\boldsymbol{a}})^H \boldsymbol{Q}^{-1} (\mathbb{E}_{G_2}[\boldsymbol{a}] - \bar{\boldsymbol{a}}) \leq \gamma _1\\
        \mathbb{E}_{G_2}[(\boldsymbol{a} - \bar{\boldsymbol{a}})(\boldsymbol{a} - \bar{\boldsymbol{a}})^H] \preceq  (1+ \gamma _2)\bar{\boldsymbol{\Sigma}}\\
    \end{array}
    \right\},
\end{equation}
}\noindent
which has been formulated in \cite[Example 8.10]{Delage2021}. The third constraint in \eqref{eq:uncertainty_a_prime_prime} allows us to control how far the realization might be from $\bar{\boldsymbol{a}}$ on average (see \cite{Delage2021}). Given \eqref{eq:uncertainty_a_prime_prime}, the minimization problem within the constraint of problem \eqref{eq:DRO-based RAB} can be reformulated as:
\begin{equation} \label{eq:constraintDRO_D2_prime_prime}
\begin{split}
    \underset{G_2 \in \mathcal{M}_2}{\text{minimize}} \quad &\int _{\mathcal{Z}_2} \boldsymbol{a}^H \boldsymbol{w} \boldsymbol{w}^H \boldsymbol{a}  \, dG_2(\boldsymbol{a})\\
    \text{subject to} \quad &\int _{\mathcal{Z}_2} dG_2(\boldsymbol{a}) = 1,\\
    &\left \| \int _{\mathcal{Z}_2} \boldsymbol{Q}^{-\frac{1}{2}} \boldsymbol{a} \, dG_2(\boldsymbol{a}) - \boldsymbol{Q}^{-\frac{1}{2}} \bar{\boldsymbol{a}} \right \| \leq \sqrt{\gamma _1},\\
    &\int _{\mathcal{Z}_2} (\boldsymbol{a} - \bar{\boldsymbol{a}})(\boldsymbol{a} - \bar{\boldsymbol{a}})^H \, dG_2(\boldsymbol{a}) \preceq (1 + \gamma _2) \bar{\boldsymbol{\Sigma}}.\\
\end{split}
\end{equation}
The Lagrangian function for problem \eqref{eq:constraintDRO_D2_prime_prime} can be formulated as follows:
\begin{equation} \label{eq:LagrangianConstraintDRO_D2_prime_prime}
\begin{split}
    \mathcal{L}_2&(G_2, x, y, \boldsymbol{z}, \boldsymbol{Z}) = \int _{\mathcal{Z}_2} \boldsymbol{a}^H \boldsymbol{w} \boldsymbol{w}^H \boldsymbol{a} \, dG_2(\boldsymbol{a})\\ 
    &+ x \left( \int _{\mathcal{Z}_2} dG_2(\boldsymbol{a}) - 1 \right)\\
    &- \Re \left(\boldsymbol{z}^H \left( \int _{\mathcal{Z}_2} \boldsymbol{Q}^{-\frac{1}{2}} \boldsymbol{a} \, dG_2(\boldsymbol{a}) - \boldsymbol{Q}^{-\frac{1}{2}} \bar{\boldsymbol{a}} \right)\right) - y \sqrt{\gamma _1}\\
    &+ \text{tr}\left(\boldsymbol{Z} \left(\int _{\mathcal{Z}_2} (\boldsymbol{a} - \bar{\boldsymbol{a}})(\boldsymbol{a} - \bar{\boldsymbol{a}})^H \, dG_2(\boldsymbol{a}) - (1 + \gamma _2) \bar{\boldsymbol{\Sigma}}\right)\right),\\
\end{split}
\end{equation}
where $x \in \mathbb{R}$, $(\boldsymbol{z}, y) \in \mathbb{C}^N \times \mathbb{R}$, $\|\boldsymbol{z}\| \leq y$, and $\boldsymbol{Z} \succeq \boldsymbol{0}$. Following similar approach as in the previous dual derivations, we formulate the dual problem for \eqref{eq:constraintDRO_D2_prime_prime} as:
\begin{equation} \label{eq:dualProblemConstraint_D2_prime_prime}
\begin{split}
    \underset{x, \boldsymbol{z}, \boldsymbol{Z}}{\text{maximize}} \;\; &-x + \Re (\bar{\boldsymbol{a}}^H \boldsymbol{Q}^{-\frac{1}{2}} \boldsymbol{z}) - \sqrt{\gamma _1} \| \boldsymbol{z} \|\\
    &\qquad \qquad \qquad \qquad \; \; - \text{tr}(\boldsymbol{Z}((1 + \gamma _2)\bar{\boldsymbol{\Sigma}} - \bar{\boldsymbol{a}} \bar{\boldsymbol{a}}^H))\\
    {\text{subject to}} \;\; &\boldsymbol{a}^H \!(\boldsymbol{w} \boldsymbol{w}^H \!\!+\! \boldsymbol{Z}) \boldsymbol{a} \!-\! 2 \Re (\boldsymbol{a}^H \!(\boldsymbol{Z} \bar{\boldsymbol{a}} \!+\! \frac{1}{2} \boldsymbol{Q}^{-\frac{1}{2}} \boldsymbol{z})) \!+\! x \!\geq\! 0,\\
    &\qquad \qquad \qquad \qquad \qquad \qquad \qquad \qquad \; \; \; \forall \boldsymbol{a} \in \mathcal{Z}_2,\\
    &\boldsymbol{Z} \succeq \boldsymbol{0}, \boldsymbol{z} \in \mathbb{C}^N, x \in \mathbb{R}.\\
\end{split} 
\end{equation}
Moreover, strong duality between \eqref{eq:constraintDRO_D2_prime_prime} and \eqref{eq:dualProblemConstraint_D2_prime_prime} is guaranteed by analogous reasoning.

For $\mathcal{Z}_2$ defined by \eqref{supportset_steeringvector}, introducing the variable transformation $\boldsymbol{x} = \frac{1}{2} \boldsymbol{Q}^{-\frac{1}{2}} \boldsymbol{z}$, problem \eqref{eq:dualProblemConstraint_D2_prime_prime} transforms into the following maximization problem:
\begin{equation} \label{eq:dualProblemConstraint_D2_prime_prime_2}
\begin{split}
    \underset{x, \boldsymbol{z}, \boldsymbol{Z}}{\text{maximize}} \;\; &-x + 2 \Re (\bar{\boldsymbol{a}}^H \boldsymbol{x}) - 2 \sqrt{\gamma _1} \| \boldsymbol{Q}^{\frac{1}{2}} \boldsymbol{x} \|\\
    &\qquad \qquad \qquad \; \; \; \: - \text{tr}(\boldsymbol{Z}((1 + \gamma _2)\bar{\boldsymbol{\Sigma}} - \bar{\boldsymbol{a}} \bar{\boldsymbol{a}}^H))\\
    {\text{subject to}} \;\; &\left[
    \begin{array}{cc}
        \boldsymbol{w} \boldsymbol{w}^H + \boldsymbol{Z} &-\boldsymbol{Z} \bar{\boldsymbol{a}} - \boldsymbol{x}\\
        - \bar{\boldsymbol{a}}^H \boldsymbol{Z} - \boldsymbol{x}^H &x\\
    \end{array}
    \right] \succeq\\
    &\hspace{-4pt}\tau _1
    \left[
    \begin{array}{cc}
        \boldsymbol{I} &\boldsymbol{0}\\
        \boldsymbol{0} &-(1+\Delta)N\\
    \end{array}
    \right] 
    + \tau _2
    \left[
    \begin{array}{cc}
        -\boldsymbol{I} &\boldsymbol{0}\\
        \boldsymbol{0} &(1-\Delta)N\\
    \end{array}
    \right],\\
    &\tau _1, \tau _2 \leq 0.\\
    &\boldsymbol{Z} \succeq \boldsymbol{0}, \boldsymbol{z} \in \mathbb{C}^N, x \in \mathbb{R}.\\
\end{split} 
\end{equation}

Thus, in this scenario, the DRO problem \eqref{eq:DRO-based RAB} can be reformulated as the following QMI problem:
\begin{equation} \label{eq:DRO-based RAB_last_D2_prime_prime}
\begin{split}
    \text{minimize} \;\; &\rho _1 \|\boldsymbol{X}\|_F + \rho _2 \|\boldsymbol{w} \boldsymbol{w}^H + \boldsymbol{X}\|_F - \text{tr}(\boldsymbol{X} \boldsymbol{S}_0)\\
    \text{subject to} \;\; &-x + 2 \Re (\bar{\boldsymbol{a}}^H \boldsymbol{x}) - \text{tr}(\boldsymbol{Z}(\bar{\boldsymbol{\Sigma}} - \bar{\boldsymbol{a}} \bar{\boldsymbol{a}}^H))\\
    &\qquad \qquad \qquad \geq 1 + 2 \sqrt{\gamma _1} \| \boldsymbol{Q}^{\frac{1}{2}} \boldsymbol{x} \| + \gamma _2 \text{tr}(\bar{\boldsymbol{\Sigma}} \boldsymbol{Z}),\\
    &\left[
    \begin{array}{cc}
        \boldsymbol{w} \boldsymbol{w}^H + \boldsymbol{Z} &-\boldsymbol{Z} \bar{\boldsymbol{a}} - \boldsymbol{x}\\
        - \bar{\boldsymbol{a}}^H \boldsymbol{Z} - \boldsymbol{x}^H &x\\
    \end{array}
    \right] \succeq\\
    &\hspace{-4pt}\tau _1
    \left[
    \begin{array}{cc}
        \boldsymbol{I} &\boldsymbol{0}\\
        \boldsymbol{0} &-(1+\Delta)N\\
    \end{array}
    \right] 
    + \tau _2
    \left[
    \begin{array}{cc}
        -\boldsymbol{I} &\boldsymbol{0}\\
        \boldsymbol{0} &(1-\Delta)N\\
    \end{array}
    \right],\\
    &\tau _1, \tau _2 \leq 0.\\
    &\boldsymbol{w},\boldsymbol{x} \in \mathbb{C}^N, \; \boldsymbol{X} \in \mathcal{H}^N, \; \boldsymbol{Z} \succeq \boldsymbol{0}, \; x \in \mathbb{R}.
\end{split}
\end{equation}


\subsection{Alternative Uncertainty Set for INC Matrix}

Let us now direct our focus to the INC matrix. We consider the following uncertainty set as an alternative to $\mathcal{D}_1$:
\begin{equation}
    \mathcal{D}_1^\prime = \left\{G_1 \in \mathcal{M}_1 \; \; \middle| \; \;
    \begin{array}{l}
        \mathbb{P}_{G_1}(\boldsymbol{R}_{\rm{i+n}} \in \mathcal{Z}_1) = 1\\
        \mathbb{E}_{G_1}[\boldsymbol{R}_{\rm{i+n}}] \succeq (1 - \rho _1)(\boldsymbol{S}_0 + \epsilon \boldsymbol{I})\\
        \mathbb{E}_{G_1}[\boldsymbol{R}_{\rm{i+n}}] \preceq (1 + \rho _1)(\boldsymbol{S}_0 + \epsilon \boldsymbol{I})\\
    \end{array}
    \right\},
\end{equation}
where the second and third constraints ensure that the expected value of the INC matrix remains in close proximity to a diagonally loaded version of its empirical mean, with the allowable deviation regulated by $\rho _1$. Specifically, they effectively compare the expected value with a diagonally loaded version of the empirical mean matrix. By doing so, it enhances robustness against the influence of the desired signal present in the training data. This elevation of the noise floor through controlled diagonal loading mitigates the adverse impact of the desired signal, thereby ensuring more reliable and stable performance in adaptive processing scenarios.

Employing the uncertainty set $\mathcal{D}_1^\prime$, the maximization problem within the objective of problem \eqref{eq:DRO-based RAB} is reformulated as follows:
\begin{equation} \label{eq:objectiveDRO_D1_prime}
\begin{split}
    \underset{G_1 \in \mathcal{M}_1}{\text{maximize}} \quad &\int _{\mathcal{Z}_1} \boldsymbol{w}^H \boldsymbol{R} \boldsymbol{w} \, dG_1(\boldsymbol{R})\\
    \text{subject to} \quad &\int _{\mathcal{Z}_1} dG_1(\boldsymbol{R}) = 1,\\
    &\int _{\mathcal{Z}_1} \boldsymbol{R} \, dG_1(\boldsymbol{R}) \succeq (1 - \rho _1)(\boldsymbol{S}_0 + \epsilon \boldsymbol{I}),\\
    &\int _{\mathcal{Z}_1} \boldsymbol{R} \, dG_1(\boldsymbol{R}) \preceq (1 + \rho _1)(\boldsymbol{S}_0 + \epsilon \boldsymbol{I}).\\
\end{split}
\end{equation}
The Lagrangian function for \eqref{eq:objectiveDRO_D1_prime} is formulated as:
\begin{equation} \label{eq:LagrangianObjectiveDRO_D1_prime}
\begin{split}
    \mathcal{L}_1&(G_1, \lambda, \boldsymbol{X}, \boldsymbol{X'}) = \int _{\mathcal{Z}_1} \boldsymbol{w}^H \boldsymbol{R} \boldsymbol{w} \, dG_1(\boldsymbol{R})\\ 
    &+ \lambda \left(\int _{\mathcal{Z}_1} dG_1(\boldsymbol{R}) - 1 \right)\\
    &+ \text{tr}\left(\boldsymbol{X}\left(\int _{\mathcal{Z}_1} \boldsymbol{R} \, dG_1(\boldsymbol{R}) - (1 - \rho _1)(\boldsymbol{S}_0 + \epsilon \boldsymbol{I}) \right)\right)\\
    &+ \text{tr}\left(\boldsymbol{X'}\left(-\int _{\mathcal{Z}_1} \boldsymbol{R} \, dG_1(\boldsymbol{R}) + (1 + \rho _1)(\boldsymbol{S}_0 + \epsilon \boldsymbol{I}) \right)\right),\\
\end{split}
\end{equation}
where $\lambda \in \mathbb{R}$ and $\boldsymbol{X}, \boldsymbol{X'} \succeq \boldsymbol{0}$. Following similar approach as in dual derivations of \eqref{eq:objectiveDRO}, the dual problem of \eqref{eq:objectiveDRO_D1_prime} can be expressed as:
\begin{equation} \label{eq:dualProblemObjectiveLast_D1_prime}
\begin{split}
    \underset{\boldsymbol{X}, \boldsymbol{X'}}{\text{minimize}} \; &\delta_{\mathcal{Z}_1}(\boldsymbol{w} \boldsymbol{w}^H + \boldsymbol{X}  - \boldsymbol{X'})\\
    & \qquad \;\; + \text{tr}((\boldsymbol{X'}(1 + \rho _1) - \boldsymbol{X}(1 - \rho _1))(\boldsymbol{S}_0 + \epsilon \boldsymbol{I}))\\
    {\text{subject to}} \;\; &\boldsymbol{X}, \boldsymbol{X'} \succeq \boldsymbol{0}\ (\in \mathcal{H}^N_+).
\end{split} 
\end{equation}
Additionally, strong duality between \eqref{eq:objectiveDRO_D1_prime} and \eqref{eq:dualProblemObjectiveLast_D1_prime} holds due to similar reasoning. 

Accordingly, using $\mathcal{D}_1^\prime$ and incorporating the support function derived in \eqref{supportfunctionTraceConstraintSet}, the DRO problem \eqref{eq:DRO-based RAB} is reformulated as follows: 
\begin{equation} \label{eq:DRO-based RAB_last_D1_prime}
\begin{split}
    \text{minimize} \;\; &\rho _2 \|\boldsymbol{w} \boldsymbol{w}^H + \boldsymbol{X} - \boldsymbol{X'}\|_F\\
    & \qquad \; \; + \text{tr}(((1 + \rho _1) \boldsymbol{X'} - (1 - \rho _1) \boldsymbol{X})(\boldsymbol{S}_0 + \epsilon \boldsymbol{I}))\\
    \text{subject to} \;\; &x + \Re (\boldsymbol{a}_0^H \boldsymbol{x}) - \gamma_1 \|\boldsymbol{x}\|\\
    & \qquad \qquad \quad \; - \text{tr}(\boldsymbol{Z}( (1+\gamma_2)\boldsymbol{\Sigma} + \boldsymbol{a}_0 \boldsymbol{a}_0^H)) \geq 1,\\
    &\left[
    \begin{array}{cc}
        \boldsymbol{w} \boldsymbol{w}^H + \boldsymbol{Z} &-\frac{\boldsymbol{x}}{2}\\
        -\frac{\boldsymbol{x}^H}{2} &-x\\
    \end{array}
    \right] \succeq\\
    &\hspace{-4pt}\tau _1
    \left[
    \begin{array}{cc}
        \boldsymbol{I} &\boldsymbol{0}\\
        \boldsymbol{0} &-(1+\Delta)N\\
    \end{array}
    \right] 
    + \tau _2
    \left[
    \begin{array}{cc}
        -\boldsymbol{I} &\boldsymbol{0}\\
        \boldsymbol{0} &(1-\Delta)N\\
    \end{array}
    \right],\\
    &\tau _1, \tau _2 \leq 0,\\
    &\boldsymbol{w},\boldsymbol{x} \in \mathbb{C}^N, \;, \; \boldsymbol{X}, \boldsymbol{X'}, \boldsymbol{Z} \succeq \boldsymbol{0}, \; x, \tau _1, \tau _2 \in \mathbb{R}.
\end{split}
\end{equation}

Evidently, any of the problems \eqref{eq:DRO-based RAB_last_D2_prime}, \eqref{eq:DRO-based RAB_last_D2_prime_prime}, or \eqref{eq:DRO-based RAB_last_D1_prime} can be solved by finding a rank-one solution for its LMI relaxation problem (like \eqref{eq:DRO-based RAB_last_relaxed}) via Algorithm~\ref{alg:alg1}.

\section{Numerical Examples} \label{Numerical Examples}

For numerical simulations, we consider a uniform linear array (ULA) consisting of $N=10$ omnidirectional sensor elements, each spaced at intervals of half a wavelength. Each sensor is subject to additive noise, modeled as a zero-mean complex Gaussian random variable, which is temporally and spatially independent from the noise in other sensors, with a power level set at $0$~dB. In all simulation examples, we assume the presence of two interfering sources, each with the interference-to-noise ratio (INR) of $30$~dB, impinging upon the array from incident angles of $-5^{\circ}$ and $15^{\circ}$. The angular sector of interest spans from $0^{\circ}$ to $10^{\circ}$. The true arrival direction of the desired signal is $5^{\circ}$, whereas the assumed direction is $\theta _0 = 1^{\circ}$. In addition to the signal look direction mismatch, we take into consideration mismatch caused also by wavefront distortion in an inhomogeneous media \cite{Khabbazibasmenj2012}, as well as potential imperfections from limited-resolution phase shifters. Specifically, we model the desired signal steering vector as being distorted by accumulated phase shifts, where the phase increments are modeled as independent Gaussian variables with zero mean and a standard deviation of $0.02$. This value reflects a moderate level of distortion appropriate for the considered setup, and remains fixed throughout each simulation run. We assume that the desired signal is always present in the training data and the training sample size $T$ is preset to $100$. The result for each scenario in our simulations is obtained by conducting 200 independent simulation runs. Optimization problems are solved using CVX \cite{diamond2016cvxpy}, \cite{grant2013cvx}.

We evaluate the performance of the proposed beamformer by comparing it with the distributional RAB problem from \cite[Eqs. (27)-(31)]{Li2018} and the distributional RAB problem from \cite[Problem (26)]{Zhang2015} across various scenarios. Unless otherwise specified, the matrix $\boldsymbol{S}_0$ in solving the RAB problem in \eqref{eq:DRO-based RAB_last_relaxed} is set to the sample data covariance matrix $\widehat{\boldsymbol{R}}$. The mean vector of the random steering vector is $\boldsymbol{a}_0 = \frac{1}{L} \sum _{l=1}^L \boldsymbol{d}(\theta _l)$, and its covariance matrix is $\boldsymbol{\Sigma} = \frac{1}{L} \sum _{l=1}^L (\boldsymbol{d}(\theta _l) - \boldsymbol{a}_0)(\boldsymbol{d}(\theta _l) - \boldsymbol{a}_0)^H$, where $\boldsymbol{d}(\theta _l)$ signifies the steering vector corresponding to $\theta _l$ and follows the structure defined by the sensor array geometry. The set $\{\theta _l\}_{l=1}^L$ represents realizations of a random variable uniformly distributed over the angular sector of interest. For the parameters in \eqref{eq:DRO-based RAB_last_relaxed}, $\rho _1$ is specified as $0.001 \|\boldsymbol{S}_0\|_F$, while $\rho _2$ is set to $1.1 \text{tr}(\boldsymbol{S}_0)$. Additionally, $\gamma_1$ and $\gamma_2$ are defined as $0.01 \|\boldsymbol{a}_0\|$ and $0.1$, respectively. The parameter $\Delta$ is also set to $0.1$. All these parameters are optimized for the corresponding methods. In \eqref{eq:DRO-based RAB_last_relaxed_rankOne}, the parameter $\alpha$ is fixed at $10^3$. The termination threshold $\eta$ in Algorithm~\ref{alg:alg1} is set to $10^{-6}$. When addressing the beamforming problem as described in \cite{Li2018}, the parameters are configured as follows: the matrix $\widehat{\boldsymbol{R}}_y$ is set equal to the sample data covariance matrix $\widehat{\boldsymbol{R}}$, and the vector $\tilde{\boldsymbol{a}}$ is equated to $\boldsymbol{d}(\theta _0)$. The threshold parameter $p$ is assigned a value of $0.9$. Furthermore, $\sigma ^2$ is computed following the approach used in the first numerical example in \cite{Li2018}, specifically employing the Gaussian mixture model. In the beamforming problem (26) of \cite{Zhang2015}, the following configurations are applied: $\widehat{\boldsymbol{R}}_y = \widehat{\boldsymbol{R}}$, $\boldsymbol{d}_l = [\Re (\boldsymbol{d}(\theta _l))^T, \Im (\boldsymbol{d}(\theta _l))^T]^T$, where $\{\theta _l\}_{l=1}^L$ represents another set of realizations uniformly distributed over the angular sector of interest. Additionally, $\hat{\boldsymbol{a}} = \frac{1}{L} \sum _{l=1}^L \boldsymbol{d}_l$, $\hat{\boldsymbol{\Sigma}} = \frac{1}{L} \sum _{l=1}^L (\boldsymbol{d}_l - \hat{\boldsymbol{a}})(\boldsymbol{d}_l - \hat{\boldsymbol{a}})^T$, $\hat{\boldsymbol{\Sigma}} = \hat{\boldsymbol{\Sigma}} + 0.1 \boldsymbol{I}$ (to ensure positive definiteness), and $\gamma = 0.1 N$. The support set $\mathcal{S}$ 
consists of three ellipsoids (i.e., $n=3$ in problem (26) of \cite{Zhang2015}): $\mathcal{S} = \mathcal{E}_1 (\hat{\boldsymbol{a}}_1, \gamma _1 \hat{\boldsymbol{\Sigma}} _1) \cup \mathcal{E}_2 (\hat{\boldsymbol{a}}_2, \gamma _2 \hat{\boldsymbol{\Sigma}} _2) \cup \mathcal{E}_3 (\hat{\boldsymbol{a}}_3, \gamma _3 \hat{\boldsymbol{\Sigma}} _3)$, with $\gamma _i = 0.2N$ for $i=1,2,3$. The entries $(\hat{\boldsymbol{a}}_i, \boldsymbol{Q}_i)$ for $i=1,2,3$, are randomly generated from a standard Gaussian distribution, with covariance matrix $\hat{\boldsymbol{\Sigma}} _i = \boldsymbol{Q}_i \boldsymbol{Q}_i ^T$.

Furthermore, the performance of the DRO-based beamformer is evaluated for alternative uncertainty sets. For the parameters of the ellipsoidal constraint in $\mathcal{D}_2^{\prime}$ and $\mathcal{D}_2^{\prime\prime}$, $\gamma _1$ is specified as $0.01 \|\boldsymbol{a}_0\|$, and $\boldsymbol{Q}$ is defined as $\boldsymbol{Q} = \boldsymbol{U} \boldsymbol{U} ^H$, where $\boldsymbol{U}$ is randomly generated from a standard complex Gaussian distribution. The parameter $\gamma _2$ is selected as $0.01 \|\boldsymbol{\Sigma}\|_F$ for $\mathcal{D}_2^{\prime}$. For $\mathcal{D}_2^{\prime\prime}$, $\gamma _2$ is set to $0.1$. For $\mathcal{D}_1^{\prime}$, $\epsilon$ is set to $0.01 \lambda_{\max}(\widehat{\boldsymbol{R}})$. Additionally, $\rho _1$ for $\mathcal{D}_1^{\prime}$ is set to $0.1$.

Fig.~\ref{fig_1} illustrates the array output SINR as a function of the per-antenna signal-to-noise ratio (SNR) with $T = 100$ snapshots. For the proposed method, the formulation used corresponds to \eqref{eq:DRO-based RAB_last_relaxed}, which employs uncertainty sets $\mathcal{D}_1$ and $\mathcal{D}_2$, and the rank-one solution is obtained through Algorithm~\ref{alg:alg1}. The proposed DRO-based beamformer demonstrates superior performance compared to the beamformers presented in \cite{Li2018} and \cite{Zhang2015} within the operational SNR range. It is important to note that, in practical scenarios, beamforming is generally not required at high SNR levels. However, beyond a certain SNR threshold, the performance of the DRO-based beamformer degrades. This degradation arises due to the increased influence of the desired signal's presence in the training data at higher SNRs. In high SNR conditions, the performance degradation can be understood by interpreting the beamforming problem as a steering vector estimation problem, as discussed in sources such as \cite{Khabbazibasmenj2012}. At elevated SNRs, the sample data covariance matrix, used as an approximation of the INC matrix, becomes increasingly influenced by the presence of the desired signal. This influence leads to inaccuracies in steering vector estimation by introducing interference directions, or linear combinations thereof, into the estimated steering vector, thus undermining the beamformer’s interference suppression capabilities. This degradation effect is especially pronounced in the DRO-based beamformer due to its explicit distributional robustness constraint on the INC matrix. In this formulation, the uncertainty set for $\boldsymbol{R}_{\rm{i+n}}$ not only captures typical interference and noise variations but also tolerates additional distributional uncertainties to protect against model mismatches. This robustness constraint means that the DRO-based model is less rigid regarding the exact structure of interference and noise, broadening the allowable variations in the estimated INC matrix. However, at high SNR, the desired signal component increasingly contributes to the sample data covariance matrix, impacting the approximation of the INC matrix by shifting its structure closer to the desired signal’s spatial characteristics. Because the DRO-based beamformer is structured to account for a wider range of interference-plus-noise profiles, it interprets these variations in the sample data covariance matrix as valid instances of interference-plus-noise covariance structures. This flexibility, while beneficial under low-SNR conditions, introduces uncertainty at high SNR, where the desired signal’s influence distorts the estimation. Consequently, interference directions or their combinations can leak into the steering vector estimate, particularly in high-SNR conditions where the DRO-based uncertainty set accommodates broader deviations. Thus, the DRO-based beamformer, while being robust against general distributional uncertainty, inadvertently increases sensitivity to the influence of the desired signal within the sample data covariance matrix, resulting in a more pronounced degradation compared to other beamformers that do not incorporate distributional robustness for the INC matrix. To mitigate this issue, we also explore the performance of our method when the nominal matrix $\boldsymbol{S}_0$ is set to a reconstructed INC matrix based on the approach in \cite{Gu2012}, rather than the sample data covariance matrix. This helps reduce the influence of the desired signal component, thereby enhancing the beamformer’s SINR performance in the high SNR regime.

Fig.~\ref{fig_2} compares the performance of the beamformers with $T = 10$ snapshots, when the exact knowledge of the INC matrix is available. As expected, the performance degradation observed at high SNR levels is eliminated. This is because the exact INC matrix ensures that steering vector estimation is not affected by distortions introduced by the desired signal component. Notably, the proposed DRO-based beamformer demonstrates superior SINR performance compared to the beamformers of \cite{Li2018} and \cite{Zhang2015} across all SNR levels, even with a reduced number of snapshots. This enhanced performance underscores the superiority of the DRO-based formulation, particularly when exact covariance information is provided, allowing to optimally suppress interference without being constrained by sample data uncertainties. Furthermore, the ability of the DRO-based beamformer to capitalize on precise INC matrix knowledge, even in low-snapshot regimes, highlights its adaptability and effectiveness in challenging beamforming scenarios.

\begin{figure}[t]
\centering
\includegraphics[width=3.49in]{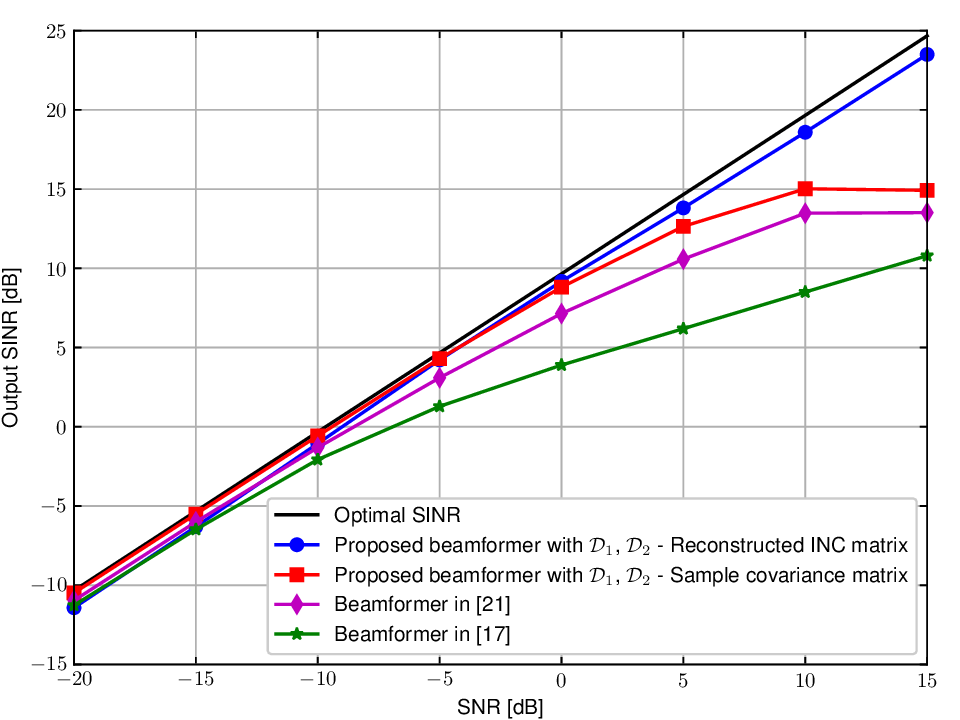}
\caption{Array output SINR versus per-antenna SNR ($T=100$ snapshots). The proposed DRO-based beamformer demonstrates superior performance compared to the beamformers of \cite{Li2018} and \cite{Zhang2015} within the operational SNR range. However, at high SNR, the performance of the DRO-based beamformer degrades due to the increased influence of the desired signal component in the sample data covariance matrix. This influence distorts the approximation of the INC matrix, impacting steering vector estimation and reducing interference suppression. Incorporating the reconstructed INC matrix from \cite{Gu2012} into the DRO-based formulation, instead of the basic sample data covariance matrix, improves performance, especially at higher SNR levels, by providing a more accurate nominal matrix for the uncertainty set.}
\label{fig_1}
\end{figure}

\begin{figure}[t]
\centering
\includegraphics[width=3.49in]{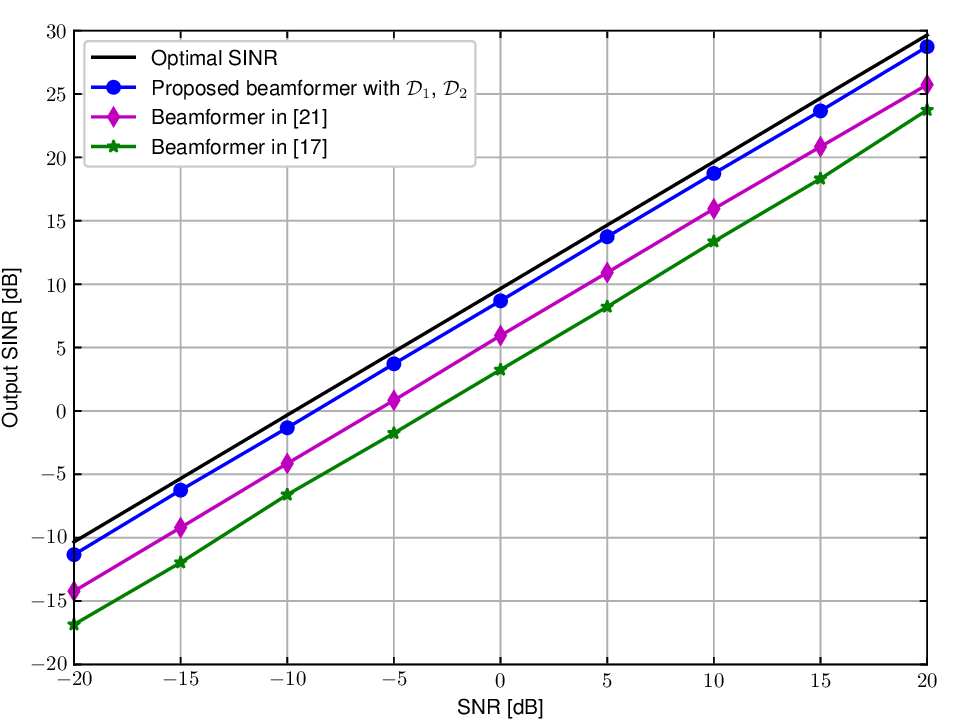}
\caption{Array output SINR versus per-antenna SNR for $T = 10$ snapshots, assuming exact knowledge of the INC matrix. There is no performance degradation at high SNR levels due to the absence of distortions caused by the desired signal component in the sample data covariance matrix. The proposed DRO-based beamformer demonstrates consistently superior SINR performance across the entire SNR range. This result highlights the benefits of exact covariance knowledge in leveraging the DRO-based formulation for interference suppression and SINR optimization.}
\label{fig_2}
\end{figure}

The performance of the proposed beamformer for different uncertainty sets, as discussed in Section~\ref{Analysis of DRO-based RAB Problem under Alternative uncertainty Sets}, is depicted in Fig.~\ref{fig_3}. For this analysis, the LMI relaxations of problems \eqref{eq:DRO-based RAB_last_D2_prime}, \eqref{eq:DRO-based RAB_last_D2_prime_prime}, and \eqref{eq:DRO-based RAB_last_D1_prime} were used, which correspond to the uncertainty sets $(\mathcal{D}_1, \mathcal{D}_2^\prime)$, $(\mathcal{D}_1, \mathcal{D}_2^{\prime\prime})$, and $(\mathcal{D}_1^\prime, \mathcal{D}_2)$, respectively. The rank-one solution was obtained using Algorithm~\ref{alg:alg1}. As observed, the performance of the proposed beamformer improves at high SNR levels when the uncertainty set $\mathcal{D}_1^\prime$ is utilized for the random INC matrix. This improvement is attributed to the fact that the uncertainty set $\mathcal{D}_1^\prime$ accounts for the presence of the desired signal in the training samples, thereby mitigating its effect. However, a slight performance degradation is noted at lower SNR values compared to the case where the uncertainty set $\mathcal{D}_1$ is employed. This degradation results from additional diagonal loading, leading to a less accurate estimation of the INC matrix. This issue can be managed by selecting an appropriate diagonal loading factor based on the SNR value. Furthermore, it is evident that performance at high SNR levels can be optimized by carefully selecting an appropriate uncertainty set for the random steering vector.

Fig.~\ref{fig_4} demonstrates the beamformer output SINR as a function of the number of snapshots, with the SNR fixed at $10$~dB. It is evident that the proposed beamformer consistently achieves better SINR compared to the other two beamformers across all snapshot counts.

\begin{figure}[t]
\centering
\includegraphics[width=3.49in]{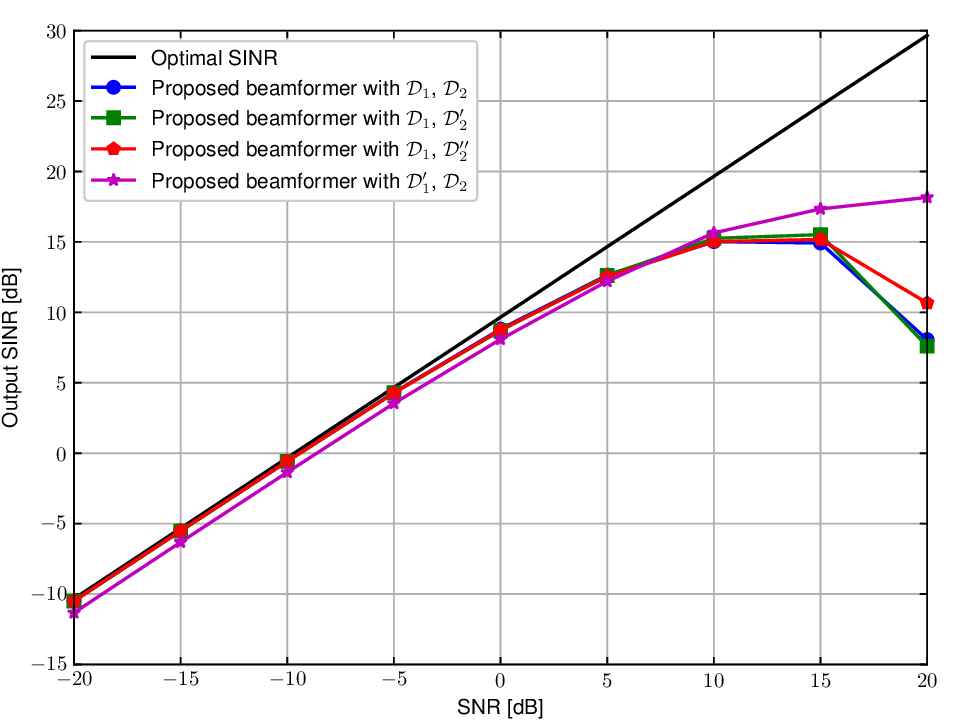}
\caption{Performance comparison of the proposed beamformer under different uncertainty sets. The performance improves at high SNR when the uncertainty set $\mathcal{D}_1^\prime$ is used for the random INC matrix. This improvement is attributed to $\mathcal{D}_1^\prime$'s ability to account for the desired signal in the training samples, thereby mitigating its influence. The set $\mathcal{D}_1^\prime$ incorporates diagonally loaded constraints that ensure the expected value of the INC matrix remains in close proximity to a diagonally loaded version of its empirical mean. This elevation of the noise floor through controlled diagonal loading mitigates the adverse impact of the desired signal, enhancing robustness.}
\label{fig_3}
\end{figure}

\begin{figure}[t]
\centering
\includegraphics[width=3.49in]{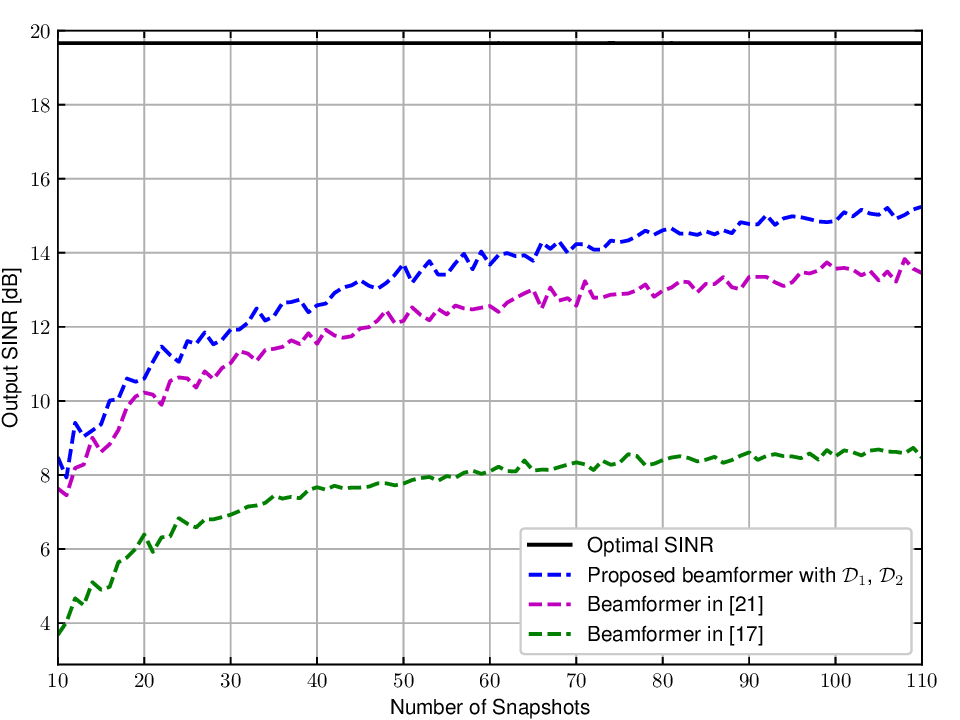}
\caption{Array output SINR versus number of snapshots with
SNR equal to $10$~dB. The proposed beamformer consistently outperforms the other beamformers across all snapshot counts, demonstrating its effectiveness in improving SINR with increasing sample size. As the number of snapshots increases, the performance gap widens, highlighting the robustness and adaptability of the proposed beamformer under realistic operating conditions.}
\label{fig_4}
\end{figure}

To further support the goal of developing a parameter-insensitive and data-driven RAB method, we examine the sensitivity of the proposed DRO-based beamformer \eqref{eq:DRO-based RAB_last_relaxed} to its parameters. Although several parameters appear in the formulation, the performance remains stable over a wide range of values of these parameters, suggesting that exact tuning is not critical. Most of these parameters, such as $\rho_1$ and $\gamma_1$, control the level of deviation allowed from nominal values and can be set based on the target application or general guidelines. To illustrate this, we evaluate the array output SINR at a fixed SNR of $-10$~dB for different values of $\rho_1$ and $\gamma_1$. The results are shown in Figs.~\ref{fig_5} and \ref{fig_6}, respectively. It can be seen that the proposed beamformer maintains consistent performance across a wide range of values, highlighting its robustness and ease of use in practical settings.

\begin{figure}[t]
\centering
\includegraphics[width=3.49in]{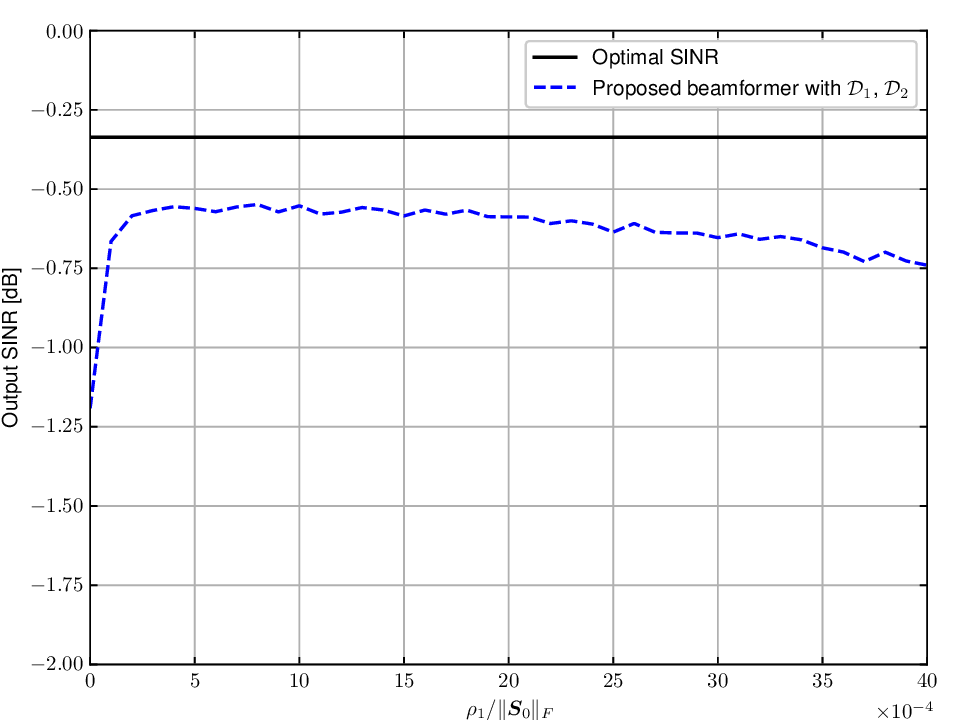}
\caption{Array output SINR versus $\rho_1 / \|\boldsymbol{S}_0\|_F$ at SNR = $-10$~dB with $T=100$ snapshots. The parameter $\rho_1$ defines the size of the uncertainty set for the INC matrix. The figure shows that the beamformer performance remains stable over a wide range of $\rho_1$ values, indicating that precise tuning is not necessary for reliable operation.}
\label{fig_5}
\end{figure}

\begin{figure}[t]
\centering
\includegraphics[width=3.49in]{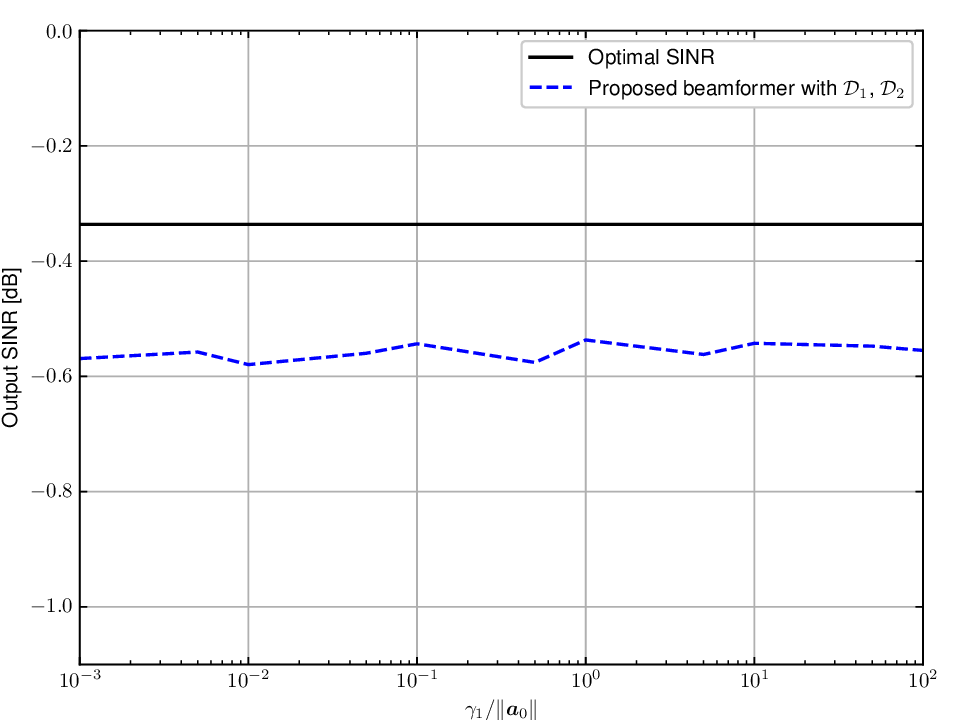}
\caption{Array output SINR versus $\gamma_1 / \|\boldsymbol{a}_0\|$ at SNR = $-10$~dB with $T=100$ snapshots. The parameter $\gamma_1$ controls the level of uncertainty in the steering vector model. The figure shows that the SINR is largely unaffected by changes in $\gamma_1$, demonstrating the robustness of the proposed method to variations in this parameter.}
\label{fig_6}
\end{figure}

A high-resolution scenario is explored by setting the number of sensors to $N = 32$. A larger number of sensors results in a more accurate estimation of the presumed direction and allows for a narrower desired angular sector due to the reduced beamwidth. In this setup, the presumed signal direction is set to $\theta_0 = 4^\circ$, with the desired sector defined as $[\theta_0 - 2^\circ, \theta_0 + 2^\circ]$. Interfering signals are placed at $0^\circ$ and $8^\circ$, which introduces interference in the vicinity of the desired sector. For our proposed DRO-based beamformer \eqref{eq:DRO-based RAB_last_relaxed}, all parameters remain unchanged from the previous scenario, except for $\rho_1$, which is now set to $\rho_1 = 0.01 \|\boldsymbol{S}_0\|_F$. We then compare the performance of our proposed method, using the reconstructed INC matrix from \cite{Gu2012} as $\boldsymbol{S}_0$, with the robust beamformer presented in \cite{Gu2012}, which is based on both reconstructed INC matrix and refined presumed steering vector.

As illustrated in Fig.~\ref{fig_7}, the proposed DRO-based beamformer \eqref{eq:DRO-based RAB_last_relaxed} performs comparably to the beamformer in \cite{Gu2012} when the reconstructed INC matrix is used as the nominal matrix and the presumed signal direction ($\theta_0 = 4^\circ$) closely approximates the true direction. It is also worth highlighting that the performance of the method in \cite{Gu2012} is highly dependent on the accuracy of the presumed direction. To demonstrate this, we repeat the simulation with $\theta_0 = 3.5^\circ$ and observe a considerable drop in the performance of the method in \cite{Gu2012}. In contrast, our proposed method does not rely on the presumed direction, and thus remains robust in the presence of such mismatches.

\begin{figure}[t]
\centering
\includegraphics[width=3.49in]{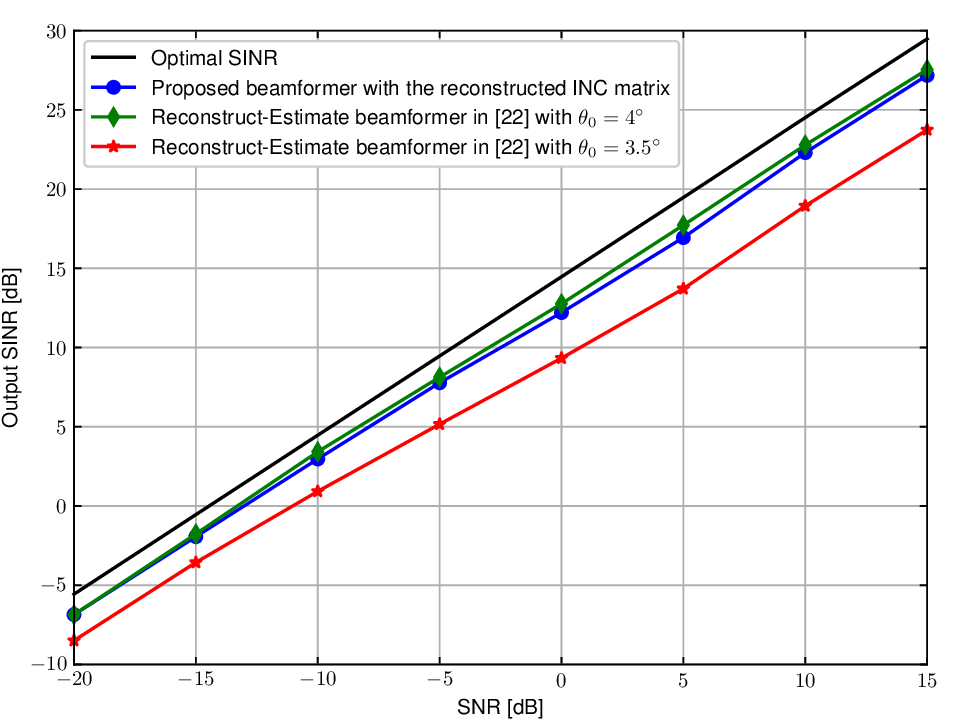}
\caption{Array output SINR versus per-antenna SNR for $N=32$ sensors with presumed direction $\theta_0 = 4^\circ$ and desired sector $[\theta_0 - 2^\circ, \theta_0 + 2^\circ]$. Two interferers are located at $0^\circ$ and $8^\circ$. The performance of the proposed DRO-based beamformer, using the reconstructed INC matrix from \cite{Gu2012} as the nominal matrix $\boldsymbol{S}_0$, is compared with the beamformer from \cite{Gu2012}. While both methods achieve comparable SINR when the presumed direction is accurate, the DRO-based beamformer maintains robustness under mismatched presumed directions, such as $\theta_0 = 3.5^\circ$, highlighting its independence from steering direction assumptions.}
\label{fig_7}
\end{figure}

\section{Conclusion} \label{Conclusion}

This paper has introduced a robust framework for adaptive beamforming utilizing distributionally robust optimization. We formulated the problem to minimize the worst-case interference-plus-noise power while ensuring the expected signal power meets a specified threshold, accounting for uncertainties in both the INC matrix and the signal steering vector. By leveraging the strong duality of linear conic programming, we effectively reformulated the DRO-based RAB problem into a QMI problem. The development of an iterative algorithm enabled us to achieve a rank-one solution for the LMI relaxation of the QMI problem. Furthermore, we establish that the stopping criterion of the proposed algorithm can be satisfied, ensuring convergence to a rank-one solution. Our evaluation across various uncertainty sets underscores the impact of uncertainty set definition on model performance. Well-defined uncertainty sets improve the effectiveness of the model. Numerical examples have validated the efficacy of our approach, demonstrating improvements in array output SINR compared to existing competitive beamforming techniques. Finally, we emphasize that the aim of this work is not to present a beamformer that outperforms all existing methods in all scenarios, but rather to introduce a novel framework based on distributional robustness. Our focus is on systematically modeling uncertainty and understanding how various uncertainty sets impact the performance and robustness of the beamformer.

\end{document}